\RequirePackage{fix-cm}

\documentclass[smallextended]{svjour3}       

\smartqed  

\usepackage{changepage}   

\usepackage{etex}
\usepackage{bbding}
\usepackage{dingbat}
\usepackage[ruled]{algorithm2e}
\usepackage{amsmath}
\usepackage{algorithmic}
\usepackage{amsfonts}
\usepackage{amsmath}
\usepackage{amssymb}
\usepackage{latexsym}
\usepackage{graphicx}
\usepackage{color}
\usepackage[normalem]{ulem}
\usepackage{balance}
\usepackage{nicefrac}
\usepackage{wasysym}
\usepackage{enumitem}
\usepackage{gensymb}
\usepackage{ctable}
\usepackage{url}
\usepackage{longtable}
\usepackage{citesort}

\usepackage{frcursive}
\usepackage{pifont,epsfig,rotating}
\usepackage{etoolbox}

\DeclareMathAlphabet{\mathpzc}{OT1}{pzc}{m}{it}

\definecolor{dgreyblue}{rgb}{0.26,0.3,0.46}             

\newcommand{\cA}{\mathcal{A}}

\newcommand{\cC}{\mathcal{C}}

\newcommand{\cE}{\mathcal{E}}
\newcommand{\cF}{\mathcal{F}}
\newcommand{\cG}{\mathcal{G}}

\newcommand{\cK}{\mathcal{K}}
\newcommand{\cL}{\mathcal{L}}
\newcommand{\cV}{\mathcal{V}}

\newcommand{\R}{{\mathbb R}}  

\renewcommand{\text}[1]{\hbox{\rm \ #1\ \/}}

\newcommand{\be}[1]{\begin{equation}\label{#1}}
\newcommand{\ee}{\end{equation}}
\newcommand{\beqn}{\begin{eqnarray*}}
\newcommand{\eeqn}{\end{eqnarray*}}
\newcommand{\beq}{\begin{eqnarray}}
\newcommand{\eeq}{\end{eqnarray}}
\newcommand{\ben}{\begin{enumerate}}
\newcommand{\een}{\end{enumerate}}
\newcommand{\bi}{\begin{itemize}}
\newcommand{\ei}{\end{itemize}}
\newcommand{\eps}{\varepsilon}

\newcommand{\IE}{{\em i.e.}\xspace}
\newcommand{\tx}{^{\rm th}}

\newtheorem{fact}{Fact}

\renewenvironment{proof}{{\noindent\bf Proof.\ }}{\hfill{\Pisymbol{pzd}{113}}\vspace{0.1in}}
\newenvironment{proof-sketch}{{\noindent\bf Sketch of Proof.\ }}{\hfill{\Pisymbol{pzd}{113}}\vspace{0.1in}}

\newcommand{\NP}{\mathsf{NP}}
\newcommand{\LP}{\mathsf{LP}}

\newcommand{\cS}{\mathcal{S}}

\newcommand{\cP}{\mathcal{P}}

\newcommand{\TB}{\vspace{-0.1ex}}\newcommand{\TiE}{\setlength{\itemsep}{-1ex}}

\newcommand{\comment}[1]{}

\newcommand{\set}[1]{\ensuremath{\{#1\}}}

\newcommand{\EG}{{\it e.g.}\xspace}
\newcommand{\FI}[1]{Fig.~\ref{#1}\xspace}

\newcommand{\iif}{{\bf{if}}}
\newcommand{\tthen}{{\bf{then}}}

\newcommand{\ffor}{{\bf{for}}}
\newcommand{\wwhile}{{\bf{while}}}

\newcommand{\ddo}{{\bf{do}}}

\newcommand{\diam}{\mathsf{diam}}

\newcommand{\cQ}{\mathcal{Q}}
\newcommand{\dist}{\mathrm{dist}}

\newcommand{\cut}{\mathrm{cut}}

\newcommand{\eopt}{ E_{\mathrm{opt}} }

\newcommand{\eqdef}{\stackrel{\mathrm{def}}{=}}
\definecolor{columbiablue}{rgb}{0.61, 0.87, 1.0}

\newcommand{\badp}{{\sc Tadp}}
\newcommand{\eadp}{{\sc Eadp}}
\newcommand{\curvgrom}{\mathfrak{C}_{\mathrm{Gromov}}}
\newcommand{\tdp}{{\sc Tdp}}

\newcommand{\dkst}{\textsc{D}$k$\textsc{S}$_3$\xspace}
\newcommand{\mnc}{{\sc Mnc}}
\newcommand{\optmnc}{{\mathsf{OPT}_{\mbox{\footnotesize\mnc}}}}
\newcommand{\vmnc}{{V_{\mbox{\footnotesize\mnc}}}}
\newcommand{\opttdp}{{\mathsf{OPT}_{\mbox{\tdp}}}}
\newcommand{\UGC}{{\sc Ugc}\xspace}
\newcommand{\ETH}{{\sc Eth}\xspace}
\newcommand{\SAT}{{\sc Sat}\xspace}
\newcommand{\kSAT}{$k$-{\sc Sat}\xspace}
\newcommand{\tSAT}{$3$-{\sc Sat}\xspace}
\newcommand{\ham}{{{\sc Cubic}-{\sc Hp}}\xspace}
\newcommand{\wte}{\widetilde{E}}
\newcommand{\whe}{\widehat{E}}
\newcommand{\ccap}{\mathrm{cap}}

\allowdisplaybreaks

\journalname{Algorithmica}

\begin{document}

\title{Why did the shape of your network change?  
\\
(On detecting network anomalies via non-local curvatures)
}

\titlerunning{Why did the shape of your network change?}

\author{
Bhaskar DasGupta$^\ast$
\and
Mano Vikash Janardhanan
\and
Farzane Yahyanejad
}

\institute{
$^\ast$Bhaskar DasGupta (corresponding author) \at
           Department of Computer Science, University of Illinois at Chicago, Chicago, IL 60607, USA \\
           Tel.: +312-355-1319\\
           Fax: +312-413-0024\\
           \email{bdasgup@uic.edu} 
\and
Mano Vikash Janardhanan \at
           Department of Mathematics, University of Illinois at Chicago, Chicago, IL 60607, USA \\
           \email{manovikashj@gmail.com}
\and
Farzaneh Yahyanejad \at
           Department of Computer Science, University of Illinois at Chicago, Chicago, IL 60607, USA \\
           \email{farzanehyahyanejad@gmail.com}
}

\date{Received: date / Accepted: date}

\maketitle

\begin{abstract}
Anomaly detection problems (also called \emph{change-point detection} problems)
have been studied in data mining, statistics and computer science over the last several decades 
(mostly in \emph{non-network context}) in applications such as 
medical condition monitoring, weather change detection and speech recognition.
In recent days, however, anomaly detection problems have become increasing more relevant in the context of 
\emph{network science} since useful insights for many complex systems in biology, finance and social science 
are often obtained by representing them via networks.
Notions of local and non-local curvatures of higher-dimensional geometric shapes and topological spaces 
play a \emph{fundamental} role in physics and mathematics 
in characterizing anomalous behaviours of these higher dimensional entities. 
However, using curvature measures to detect anomalies in networks 
is not yet very common.
To this end, a main goal in this paper to formulate and analyze curvature analysis methods to provide the foundations of 
systematic approaches to find \emph{critical components} and \emph{detect anomalies} in networks.
For this purpose, we use two measures of network curvatures
which depend on non-trivial global properties, such as distributions of geodesics and 
higher-order correlations among nodes, of the given network.
Based on these measures, we precisely formulate several computational problems related to anomaly detection in 
static or dynamic networks, and provide non-trivial computational complexity results 
for these problems. 
This paper must \emph{not} be viewed as delivering the final word 
on appropriateness and suitability of specific curvature measures.
Instead, it is our hope that this paper will stimulate and motivate further theoretical or empirical research
concerning the exciting interplay between notions of 
curvatures from network and non-network domains, a \emph{much} desired goal in our opinion.
\end{abstract}

\keywords{Anomaly detection \and Gromov-hyperbolic curvature \and geometric curvature \and exact and approximation algorithms
\and inapproximability}
\PACS{02.10.Ox \and 89.20.Ff \and 02.40.Pc}
\subclass{MSC 68Q25 \and MSC 68W25 \and MSC 68W40 \and MSC 05C85}

\section{Introduction}
\label{sec-intro}

Useful insights for many complex systems are often obtained by representing them as networks and 
analyzing them using 
graph-theoretic and combinatorial 
algorithmic tools~\cite{DL16,Newman-book,Albert-Barabasi-2002}.
In principle, we can classify these networks into \emph{two} major classes: 
\begin{enumerate}[label=$\triangleright$,leftmargin=0.8cm]
\item
\emph{Static} networks that model the corresponding system by \emph{one} fixed network.
Examples of such networks include 
biological signal transduction networks \emph{without} node dynamics, 
and many social networks.
\item
\emph{Dynamic} networks where elementary components of the 
network (such as nodes or edges) are added and/or removed as the network \emph{evolves} over time.
Examples of such networks include 
biological signal transduction networks \emph{with} node dynamics, 
causal networks 
reconstructed from DNA microarray time-series data,
biochemical reaction networks 
and dynamic social networks.
\end{enumerate}
Typically, such networks may have so-called \emph{critical} 
(\emph{elementary}) components 
whose presence or absence alters 
some significant \emph{non-trivial non-local} 
property\footnote{A non-trivial property usually refers to a property such that 
a significant percentage of all possible networks satisfies the property and
also a significant percentage of all possible networks does \emph{not} satisfy the property.
A non-local property (also called global property) usually refers to a property that \emph{cannot} be inferred 
by simply looking at a \emph{local} neighborhood of any one node.} 
of these networks. 
For example:
\begin{enumerate}[label=$\triangleright$,leftmargin=0.8cm]
\item
For a static network, 
there is a rich history in finding various types of 
critical 
components 
dating back to quantifications of fault-tolerance or redundancy
in electronic circuits or routing networks.
Recent examples of practical application of determining critical and non-critical components
in the context of systems biology 
include quantifying redundancies in biological networks~\cite{KW96,TSE99,ADGGHPSS11}
and confirming the existence of central influential neighborhoods in biological networks~\cite{ADM14}.
\item
For a dynamic network, 
critical components may correspond to a set of nodes or edges whose addition and/or removal 
\emph{between} two time steps
alters a significant topological property (\EG, connectivity, average degree) of the network. 
Popularly also known as the \emph{anomaly detection} or \emph{change-point detection}~\cite{AC17,KS09} problem,
these types of problems have also been studied over the last several decades in 
data mining, statistics and computer science \emph{mostly in the ``non-network'' context of time series data}
with applications to areas such as 
medical condition monitoring~\cite{Yang06,Bosc03}, weather change detection~\cite{Ducre03,Reev07}
and speech recognition~\cite{Chow11,Rybach09}.
\end{enumerate}
In this paper we seek to address research questions of the following \emph{generic} nature:
\begin{quote}
\emph{``Given a static or dynamic network, identify the critical
components of the network that ``encode'' significant non-trivial global properties of the network''}.
\end{quote}
To identify critical components, one first needs to provide details for following four specific items:
\smallskip
\begin{adjustwidth}{0.6cm}{}
\begin{description}
\item[(\emph{i})]
network model selection,
\item[(\emph{ii})]
network evolution rule for dynamic networks,
\item[(\emph{iii})]
definition of elementary critical components, and 
\item[(\emph{iv})]
network property selection (\IE, the global properties of the network to be investigated).
\end{description}
\end{adjustwidth}
\smallskip
The specific details for these items for this paper are as follows:
\medskip
%
\begin{adjustwidth}{0.6cm}{}
\begin{description}
\item[(\emph{i}) Network model selection:]
Our network model will be undirected graphs.
\item[(\emph{ii}) Network evolution rule for dynamic networks:]
Our dynamic networks follow the time series model and 
are given as a sequence of networks over \emph{discrete} time steps, where 
each network is obtained from the previous one in the sequence by adding and/or deleting some nodes and/or edges. 
\item[(\emph{iii}) Critical component definition:]
Individual edges are elementary members of critical components.
\item[(\emph{iv}) Network property selection:]
The network measure for this paper will be based on one or more well-justified notions of ``network curvature''. 
More specifically, we will use  
\textbf{(\emph{a})}
\textbf{Gromov-hyperbolic curvature} based on the properties of \emph{exact and approximate} 
geodesics distributions and higher-order connectivities
and 
\textbf{(\emph{b})}
\textbf{geometric curvatures} based on 
identifying network motifs with \emph{geometric complexes} (``geometric motifs'' in systems biology jargon) 
and then using Forman's combinatorializations.
\end{description}
\end{adjustwidth}

\subsection{Organization of the paper and a summary of our contributions}

The rest of the paper is organized as described below.

In Section{~\ref{sec-defn}} we introduce some 
basic definitions and notations and provide a summary list of 
some other notations that are are used throughout the rest of the paper.

In Section{~\ref{sec-just-etc}} we discuss the relevant 
background, motivation, and justification for using the curvature measures and provide two illustrative examples
in which curvature measures detect anomaly where other simpler measures do not.
We also remark on the limitations of our theoretical results that may be useful to future researchers.

In Section~\ref{def-curv} we define and motivate the two notions of graph curvature that is used in this paper
in the following manner: 
\begin{enumerate}[label=$\triangleright$,leftmargin=0.8cm]
\item
The Gromov-hyperbolic curvature is introduced in Section~\ref{def-grom-curv} together with 
justifications for using them, relevant known results and some clarifying remarks about them.
\item
Generic notions of geometric curvatures are introduced in Section~\ref{sec-geom-defn} together with 
relevant topological concepts necessary to define them and justifications for using them.
The precise definition of the geometric curvature used in this paper is given by 
Equation~\eqref{gmeq1} in Section~\ref{sec-gmeq1}.
\end{enumerate}
In Section~\ref{def-prob} we present our formalizations of anomaly detection problems on networks
based on curvature measures. We distinguish two types of anomaly detection problems in the following manner:
\begin{enumerate}[label=$\triangleright$,leftmargin=0.8cm]
\item
In Section~\ref{sec-static} we formalize the \emph{Extremal Anomaly Detection Problem} 
(Problem~\eadp$_{\mathfrak{C}}(G,\wte,\gamma)$)
for ``static networks'' that do \emph{not} change over time. 
\item
In Section~\ref{sec-dynamic}
we formalize the \emph{Targeted Anomaly Detection Problem} (\badp$_{\mathfrak{C}}(G_1,G_2)$)
for ``dynamic networks'' that \emph{do} change over time. 
\end{enumerate}
In Section~\ref{sec-proof-extreme} we present our results 
regarding the computational complexity of extremal anomaly detection problems for the two 
types of curvatures in the following manner:
\begin{enumerate}[label=$\triangleright$,leftmargin=0.8cm]
\item
Theorem~\ref{ext-thm} in Section~\ref{sec-proof-extreme-geometric}
states the computational complexity results for geometric curvatures.
Some relevant comments regarding Theorem~\ref{ext-thm} and an informal overview of its proof techniques
appear in Section~\ref{informal-ext-thm}, whereas the precise technical proofs for 
Theorem~\ref{ext-thm} are presented separately in Section~\ref{sec-app-proof-ext-thm}.
\item
Theorem~\ref{thm-eadp-curvgrom} in Section~\ref{sec-proof-extreme-gromov}
states the computational complexity results for Gromov-hyperbolic curvature.
An informal overview of the proof techniques for Theorem~\ref{thm-eadp-curvgrom} 
appears in the very beginning of Section~\ref{sec-app-proof-thm-eadp-curvgrom}, whereas the precise technical proofs for 
Theorem~\ref{thm-eadp-curvgrom} are presented in the remaining part of 
the same section.
\end{enumerate}
In Section~\ref{sec-proof-target} we present our results 
regarding the computational complexity of targeted anomaly detection problems for the two 
types of curvatures in the following manner:
\begin{enumerate}[label=$\triangleright$,leftmargin=0.8cm]
\item
Theorem~\ref{thm-badp} in Section~\ref{sec-proof-target-geom}
states the computational complexity results for geometric curvatures.
An informal overview of the proof techniques for Theorem~\ref{thm-badp} 
appears in Section~\ref{sec-informal-badp}, whereas the precise technical proofs for 
Theorem~\ref{thm-badp} are presented separately in Section~\ref{sec-app-proof-thm-badp}.
\item
Theorem~\ref{thm-grom-hard} in Section~\ref{sec-badp-all}
states the computational complexity results for Gromov-hyperbolic curvature.
Some relevant comments regarding Theorem~\ref{thm-grom-hard} and an informal overview of its proof techniques
appear in Section~\ref{sec-informal-badp-gromov}, whereas the precise technical proofs for 
Theorem~\ref{thm-grom-hard} are presented separately in Section~\ref{sec-app-proof-thm-grom-hard}.
\end{enumerate}
Finally, we conclude in Section~\ref{sec-conclusion} with a few interesting research problems 
for future research.

\medskip
\noindent
\textbf{Remarks on the organization of our proofs} 

\medskip
Many of our proofs in Sections~\ref{sec-proof-extreme}--\ref{sec-proof-target} are long, are complicated or 
involve tedious calculations. 
For easier understanding and to make the paper more readable, when appropriate we have included 
a subsection generically titled 
``Proof techniques and relevant comments regarding Theorem $\dots\dots$''
before providing the actual detailed proofs. 
The reader \emph{is} cautioned however that these brief subsections are meant to provide some general idea
and subtle points behind the proofs and should \emph{not} be considered as a substitution for more formal proofs.

\section{Basic definitions and notations}
\label{sec-defn}
For an undirected unweighted graph $G=(V,E)$ of $n$ nodes $v_1,\dots,v_n$, 
the following notations related to $G$ are used throughout:
\begin{enumerate}[label=$\blacktriangleright$,leftmargin=0.8cm]
\item
$v_{i_1}\leftrightarrow v_{i_2}\leftrightarrow v_{i_3}\leftrightarrow \dots\leftrightarrow v_{i_{k-1}}\leftrightarrow v_{i_k}$ 
denotes a path of \emph{length} $k-1$ consisting of the edges 
$\set{v_{i_1},v_{i_2}}$, 
$\set{v_{i_2},v_{i_3}}$,
$\dots$,
$\set{v_{i_{k-1}},v_{i_k}}$.
\item
$\overline{u,v}$ and $\dist_G(u,v)$ 
denote a \emph{shortest path} and the distance 
(\IE, number of edges in $\overline{u,v}$)
between nodes $u$ and $v$, respectively.
\item
$\diam(G)=\max_{v_i,v_j} \set{ \dist_G(v_i,v_j) }$
denotes the \emph{diameter} of $G$.
\item
$G\setminus E'$
denotes the graph obtained from $G$ 
by removing 
the edges in $E'$ from $E$.
\end{enumerate}
A $\eps$-approximate solution (or simply an $\eps$-approximation) 
of a 
minimization (\emph{resp}., maximization) problem is a
solution with an objective value no larger than (\emph{resp}., no smaller than) $\eps$ times 
(\emph{resp}., $\nicefrac{1}{\eps}$ times) the
value of the optimum; an algorithm of {\em performance} or 
{\em approximation ratio} $\eps$ produces an $\eps$-approximate solution.
A problem is $\eps$-\emph{inapproximable} under a certain complexity-theoretic assumption means that the problem does not 
admit a polynomial-time
$\eps$-approximation algorithm assuming that the complexity-theoretic assumption is true.
We will also use other \emph{standard} definitions from structural complexity theory 
as readily available in any graduate level textbook on algorithms such as~\cite{V01}.


Other specialized notations used in the paper are defined when they are first needed. 
For the benefit of the reader, we provide a list of some such commonly used notations in the paper 
with brief comments about them in Table~\ref{tab-not}. Please see the referring section 
for exact descriptions of these notations.

\begin{table}[htbp]
\begin{tabular}{l l l}
\toprule
  {\textbf{Nomenclature or brief explanation}}&  
       {\textbf{Notation}}	& 
           \hspace*{-0.2in}
	        { 
					  \begin{tabular}{l}
					   \textbf{Referring}
						 \\
						 \textbf{section}
					  \end{tabular}
					}
\\
\midrule
path of length $k-1$  
     & 
$v_{i_1}\leftrightarrow \dots\leftrightarrow v_{i_k}$
     & 
Section~\ref{sec-defn}
\\
\midrule
shortest path, distance between nodes $u$ and $v$
     & 
$\overline{u,v}$, $\dist_G(u,v)$ 
     & 
Section~\ref{sec-defn}
\\
\midrule
diameter of graph $G$
     & 
$\diam(G)$
     & 
Section~\ref{sec-defn}
\\
\midrule
the graph $(V,E\setminus E')$ where $G=(V,E)$
     & 
$G\setminus E'$
     & 
Section~\ref{sec-defn}
\\
\midrule
curvature of graph $G$ 
     & 
$\mathfrak{C}$ or $\mathfrak{C}(G)$
     & 
Section~\ref{def-curv}
\\
\midrule
geodesic triangle 
     & 
$\Delta_{u,v,w}$
     & 
Section~\ref{def-grom-curv}
\\
\midrule
Gromov-hyperbolic curvature of $\Delta_{u,v,w}$
     & 
$\curvgrom(\Delta_{u,v,w})$
     & 
Section~\ref{def-grom-curv}
\\
\midrule
Gromov-hyperbolic curvature of graph $G$ 
     & 
$\curvgrom$ or $\curvgrom(G)$ 
     & 
Section~\ref{def-grom-curv}
\\
\midrule
$k$-simplex of $k+1$ affinely independent points
     & 
$\cS\big(x_0,\dots,x_k\big)$
     & 
Section~\ref{sec-app-topo}
\\
\midrule
order $d$ association of $p$-face $f^p$ of a $q$-simplex
     & 
$f_d^p$
     & 
Section~\ref{sec-gmeq1}
\\
\midrule
geometric curvature of graph $G$
     & 
$\mathfrak{C}^p_d$ or $\mathfrak{C}^p_d(G)$
     & 
Section~\ref{sec-gmeq1}
\\
\midrule
\hspace*{-0.1in}
\begin{tabular}{l}
Extremal Anomaly Detection Problem,
\\
value of its optimal solution 
\end{tabular}
     & 
\hspace*{-0.1in}
\begin{tabular}{l}
\eadp$_{\mathfrak{C}}(G,\wte,\gamma)$, 
\\
$\mathsf{OPT}_{\mbox{\eadp}_{\mathfrak{C}}}(G,\wte,\gamma)$
\end{tabular}
     & 
Section~\ref{sec-static}
\\
\midrule
\hspace*{-0.1in}
\begin{tabular}{l}
Targeted Anomaly Detection Problem,
\\
value of its optimal solution 
\end{tabular}
     & 
\hspace*{-0.1in}
\begin{tabular}{l}
\badp$_{\mathfrak{C}}(G_1,G_2)$,
\\
$\mathsf{OPT}_{\mbox{\badp}_{\mathfrak{C}}}(G_1,G_2)$
\end{tabular}
     & 
Section~\ref{sec-dynamic}
\\
\midrule
densest-$k$-subgraph problem
     & 
\dkst 
     & 
Section~\ref{sec-app-proof-ext-thm}
\\
\midrule
\hspace*{-0.1in}
\begin{tabular}{l}
minimum node cover problem, 
\\
cardinality of its optimal solution
\end{tabular}
     & 
\mnc, $\optmnc$
     & 
Section~\ref{sec-app-proof-thm-badp}
\\
\midrule
\hspace*{-0.1in}
\begin{tabular}{l}
triangle deletion problem, 
\\
cardinality of its optimal solution
\end{tabular}
     & 
\tdp, $\opttdp$
     & 
Section~\ref{sec-app-proof-thm-badp}
\\
\midrule
Hamiltonian path problem for cubic graphs
     & 
\ham
     & 
Section~\ref{sec-informal-badp-gromov}
\\
\bottomrule
\end{tabular}
\caption{\label{tab-not}A list of some frequently used notations with brief explanations.}
\end{table}

\section{Background, motivation, justification and illustrative examples}
\label{sec-just-etc}

The main purpose of this section is to (somewhat informally) explain to the reader the appropriateness of our curvature 
measures both from a theoretical and an empirical point of view.
We also provide brief comments  on the limitations of our theoretical results which may be of use to future researchers.

\subsection{Justifications for using network curvature measures}

Prior researchers have proposed and evaluated a number of 
established network measures such as \emph{degree-based measures} (\EG, degree distribution), 
\emph{connectivity-based measures} (\EG, clustering coefficient), 
\emph{geodesic-based measures} (\EG, betweenness centrality)
and other more novel network measures~\cite{rich-club,LM07,ADGGHPSS11,bassett-et-al-2011}
for analyzing networks.
The network measures considered in this paper are ``appropriate notions'' of \emph{network curvatures}.
As demonstrated in published research works such as~\cite{ADM14,WJS16,WSJ16,SSG18},
these network curvature measures 
saliently encode \emph{non-trivial higher-order correlation} among nodes and edges that
\emph{cannot} be obtained by other popular network measures.
Some important characteristics of these curvature measures that we consider are~\cite[Section~(III)]{ADM14}\cite{JLBB11}:
\begin{enumerate}[label=$\blacktriangleright$,leftmargin=0.8cm]
\item
These curvature measures depend on \emph{non-trivial global} network properties, 
as opposed to measures such as 
\emph{degree distributions} or \emph{clustering coefficients} that are \emph{local} in nature or 
\emph{dense subgraphs} that use \emph{only} pairwise correlations.
\item
These curvature measures can mostly be computed efficiently in polynomial time,
as opposed to $\NP$-complete measures such as 
\emph{cliques}~\cite{GJ79}, 
\emph{densest}-$k$-\emph{subgraphs}~\cite{GJ79}, or 
some types of community decompositions such as \emph{modularity maximization}~\cite{DD13}. 
\item
When applied to real-world networks, 
these curvature measures can explain many phenomena one frequently encounters in real network applications
that are \emph{not} easily explained by other measures such as:
\begin{enumerate}[label=$\blacktriangleright$,leftmargin=0.8cm]
\item
paths mediating up- or down-regulation of a target node starting from the same regulator node in 
\emph{biological regulatory networks} often have many small crosstalk paths, and 
\item
existence of congestions in a node that is not a hub in \emph{traffic networks}.
\end{enumerate}
Further details about the suitability of our curvature measures for real biological or social
networks are provided in Section~\ref{sec-real-grom} for Gromov-hyperbolic curvature and 
at the end of Section~\ref{sec-gmeq1}
for geometric curvatures.
\end{enumerate}
Curvatures are very natural measures of anomaly of higher dimensional objects in 
mainstream physics and mathematics~\cite{book,Berger12}.
However, networks are \emph{discrete objects} that do \emph{not} necessarily have an associated natural geometric embedding. 
Our paper seeks to adapt the definition of curvature from the non-network domains
(\EG, from continuous metric spaces or from higher-dimensional geometric objects)
in a suitable way for detecting network anomalies. 
For example, in networks with sufficiently small Gromov-hyperbolicity and sufficiently large diameter 
a suitably small subset of nodes or edges can be removed to stretch the geodesics between two distinct parts 
of the network by an exponential amount.
Curiously this kind of property can be shown to have extreme implications on the expansion properties of such networks~\cite{a1,DKMY18},
akin to the characterization of 
singularities (an extreme anomaly) by geodesic incompleteness (\IE, stretching all geodesics passing through the region 
infinitely)~\cite{HP96}.

\subsection{Justifications for investigating the edge-deletion model}

In this paper we add or delete edges from a network while keeping the node set the same.
This scenario captures a wide variety of applications 
such as inducing desired outcomes in disease-related biological networks via gene knockout~\cite{SWLXLAA11,ZA15},
inference of minimal biological networks from indirect experimental evidences or gene 
perturbation data~\cite{ADDKSZW07,ADDS,W02}, and finding influential nodes in social and biological networks~\cite{ADGGHPSS11},
to name a few.
However, the node addition/deletion model or a mixture of node/edge addition/deletion model is also 
significant in many applications; we leave investigations of these models 
as future research topics.

\subsection{Two illustrative examples}
\label{sec-ex2}

It is obviously practically impossible to compare our curvatures measures for anomaly detection with respect to
\emph{every possible} other network measure that has been used in prior research works.
However, we do still provide two illustrative examples of comparing our curvature measures to the well-known
\emph{densest subgraph measure}.
The densest subgraph measure is defined as follows. 

\begin{definition}[Densest subgraph measure]
Given a graph $G=(V,E)$, the densest subgraph measure find a 
subgraph $(S,E_S)$ induced by a subset of nodes $\emptyset\subset S\subseteq V$ that maximizes the ratio (density) 
$\rho(S)\eqdef\frac{|E_S|}{|S|}$. 
Let $\rho(G)\eqdef\max_{\emptyset \subset S\subseteq V} \{ \rho(S) \}$ denote the 
density of a densest subgraph of $G$.
\end{definition}

An efficient polynomial time algorithm to compute $\rho(G)$ using a max-flow technique
was first provided by Goldberg~\cite{Gold84}.
We urge the readers to review the definitions of the relevant curvature measures (in Section~\ref{def-curv})
and the anomaly detection problems (in Section~\ref{def-prob})
in case of any confusion regarding the examples we provide.

\bigskip
\noindent
\textbf{Extremal anomaly detection for a static network} 
\bigskip

Consider the extremal anomaly detection problem 
(Problem~\eadp\ in Section~\ref{sec-static})
for a network $G=(V,E)$ of $10$ nodes and $20$ edges as shown in 
\FI{ex2-fig}
using the geometric curvature $\mathfrak{C}^2_3$
as defined by Equation~\eqref{gmeq1}.
It can be easily verified that 
$\mathfrak{C}^2_3(G)=6$
and
$\rho(G)=\nicefrac{9}{4}$.
Let $\wte=E$ and 
suppose that we set our targeted decrease of the curvature or density value to 
be $75\%$ of the original value, \IE, 
we set 
$\gamma=\nicefrac{3}{4}\times\mathfrak{C}^2_3(G)=\nicefrac{9}{2}$ for the geometric curvature measure and 
$\gamma=\nicefrac{3}{4}\times\rho(G)=\nicefrac{27}{16}$ for the densest subgraph measure.
It is easily verified that 
$\mathfrak{C}^2_3(G\setminus \{e_1\})=1<\nicefrac{9}{2}$,
thus showing 
$\mathsf{OPT}_{\mbox{\eadp}_{\mathfrak{C^2_3}}}(G,\wte,\gamma)=1$.
However, one can verify that more than $4$ edges will need to be deleted from $G$
to bring down the value of $\rho(G)$ to $\nicefrac{27}{16}$ in the following manner: 
since the densest subgraph in $G$ is induced by $8$ nodes and $18$ edges, if no more than 
$4$ edges are deleted then the density of this subgraph in the new graph is at least 
$\nicefrac{14}{8}>\nicefrac{27}{16}$.

\begin{figure}[htbp]
\centerline{\includegraphics[scale=0.95]{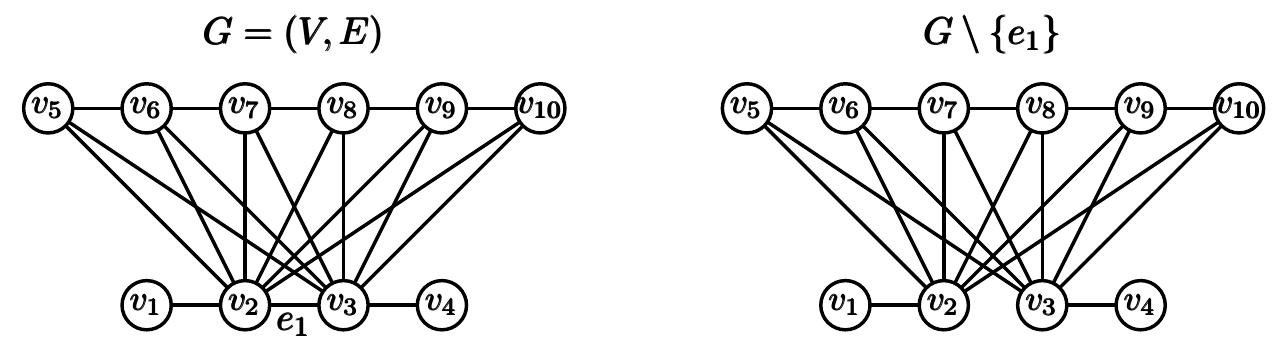}}
\caption{\label{ex2-fig}Toy example of extremal anomaly detection discussed in 
Section~\ref{sec-ex2}.
The given graph $G=(V,E)$ has $|V|=10$ nodes, $|E|=20$ edges and $\alpha=16$ triangles ($3$-cycles) giving 
$\mathfrak{C}^2_3(G)=|V|-|E|+\alpha=6$, where the densest subgraph of $G$ is the subgraph node-induced 
by the nodes $V\setminus \{v_1,v_4\}$ with $8$ nodes and $18$ edges giving 
$\rho(G)=\nicefrac{18}{8}=\nicefrac{9}{4}$.
The graph $G\setminus \{e_1\}$ has $|E\setminus\{e_1\}|=19$ edges and 
$\alpha'=10$ triangles
giving
$\mathfrak{C}^2_3(G\setminus \{e_1\})=|V|-|E\setminus\{e_1\}|+\alpha'=1$.
However, it can be verified that 
more than four edges will need to be deleted from $G$
to bring down the value of $\rho(G)$ to at most $\nicefrac{27}{16}$.
}
\end{figure}

\bigskip
\noindent
\textbf{Targeted anomaly detection for a dynamic biological network} 
\bigskip

\begin{figure}[htbp]
\centerline{\includegraphics[scale=0.62]{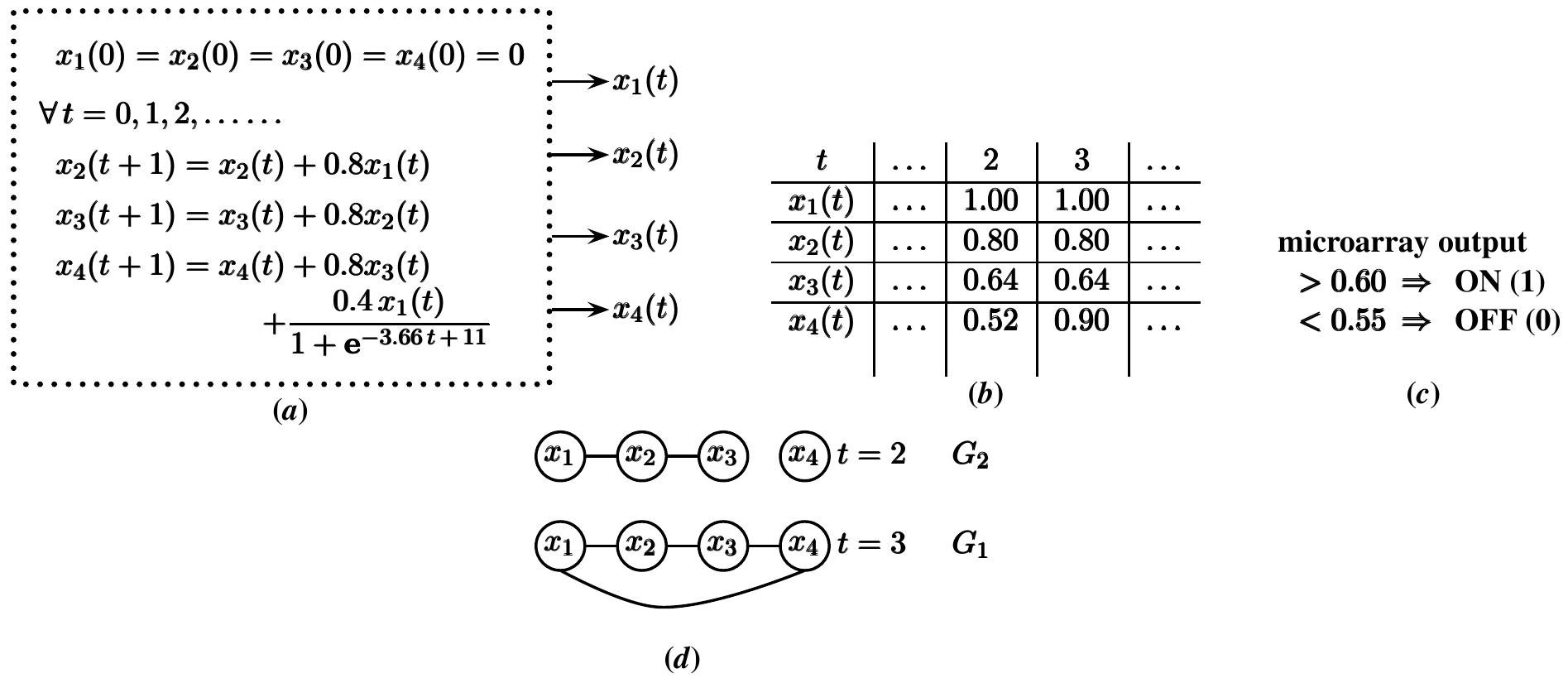}}
\caption{\label{ex1-fig}Toy example of targeted anomaly detection discussed in 
Section~\ref{sec-ex2}.
\textbf{(\emph{a})}
The original dynamical system with four observable output variables.
\textbf{(\emph{b})}
Values of the four output variable over time starting with an all-zero initial condition.
\textbf{(\emph{c})}
The thresholding rules employed to binarize the expression values in 
\textbf{(\emph{b})}
as captured by a DNA microarray.
\textbf{(\emph{d})}
The reverse-engineered network for two successive time steps.
}
\end{figure}

\noindent
Consider the targeted anomaly detection problem 
(Problem~\badp\ in Section~\ref{sec-dynamic})
using the Gromov-hyperbolic curvature (Definition~\ref{def-hyperbolic-1}).
Suppose that we have a biological dynamical system of $4$ variables $x_1,x_2,x_3,x_4$
generated by a set of recurrence equations 
as shown 
in \FI{ex1-fig}$\!$(\emph{a})
for 
$x_1(t),x_2(t),x_3(t)$ and $x_4(t)$ as a function of discrete time $t=0,1,2,3,\dots\dots$,
with the initial condition of 
$x_1(0)=x_2(0)=x_3(0)=x_4(0)=0$.
Note that 
in this biological system 
any change in the value of $x_1$ affects $x_4$ with a delay.
These recurrence equations are not known to the observer, but
they generate a sequence of real values of the state variables for each successive discrete
time units (shown in \FI{ex1-fig}$\!$(\emph{b}) for $t=2$ and $t=3$).
Suppose that an observer measures a binarized version of these real values of the state variables for each successive discrete
time units using a DNA microarray by using thresholds as shown in \FI{ex1-fig}$\!$(\emph{c}),
and then reverse-engineers a time-varying 
network by using the hitting-set approach of Krupa~\cite[Section 5.4.2]{DL16}\cite{Jarrah}
with a time delay of $2$ (the corresponding network for $t=2$ and $t=3$ is shown in 
\FI{ex1-fig}$\!$(\emph{d})).
Suppose that for our targeted anomaly detection problem 
we fix our attention to the two graphs $G_2$ and $G_1$ constructed in the
two successive time steps $t=2$ and $t=3$, respectively, 
where $G_2=G_1\setminus \big\{ \{x_1,x_4\},\{x_3,x_4\} \big\}$ is the target graph. 
It can be easily verified that 
$\curvgrom(G_1)=\rho(G_1)=1$, 
$\curvgrom(G_2)=0$, and 
$\rho(G_2)=\nicefrac{1}{2}$.
Since 
$\curvgrom(G_1\setminus \{x_1,x_4\})=0$ 
it follows that 
we only need to delete the edge 
$\{x_1,x_4\}$ 
to bring down the value of $\curvgrom(G_1)$ to $\curvgrom(G_2)$.
However, both the edges 
$\{x_1,x_4\}$ and $\{x_3,x_4\}$
need to be deleted from $G_1$ 
to bring down the value of $\rho(G_1)$ to $\rho(G_2)$.

\subsection{Brief remarks regarding the limitations of our theoretical results}

Our theoretical results obviously have some limitations, specially for real-world networks.
For example, 
our inapproximability results for the Gromov-hyperbolic curvature require a 
high average node degree. Thus, for real-world networks such as scale-free networks the 
inapproximability bounds may not apply. 
On another note, for geometric curvatures we only considered the first-order non-trivial measure 
$\mathfrak{C}^2_d$, but perhaps more salient non-trivial topological properties could be captured by
using $\mathfrak{C}^p_d$ for $p>2$.

\section{Two notions of graph curvature}
\label{def-curv}

For this paper, a \emph{curvature} for a graph $G$ is a 
function 
$\mathfrak{C} \eqdef \mathfrak{C}(G) : G\mapsto\R$. 
There are several ways in which 
network curvature can be defined depending on the type of global properties 
the measure is desired to affect; in this paper we consider two such definitions as described subsequently.

\subsection{Gromov-hyperbolic curvature}
\label{def-grom-curv}

This measure for a metric space was first suggested by Gromov in
a group theoretic context~\cite{G87}.
The measure was first defined for \emph{infinite} continuous metric space~\cite{book},
but was later also adopted for \emph{finite} graphs.
Usually the measure is defined via \emph{geodesic triangles} as stated in Definition~\ref{def-hyperbolic-1}.
For this definition, it would be useful to consider the given graph $G$ as a metric graph, \IE, 
we identify (by an isometry) any edge $\{u,v\}\in E$ with
the real interval $[0,1]$ and thus any point in the interior of the edge $\{u,v\}$ can also be thought as
a (virtual) node of $G$.
Define a geodesic triangle $\Delta_{u,v,w}$ to be 
an ordered triple of three shortest paths 
$(\overline{u,v}$, $\overline{u,w}$ and $\overline{v,w})$ for the three nodes $u,v,w$ in $G$.

\begin{definition}[Gromov-hyperbolic curvature measure via geodesic triangles]
\label{def-hyperbolic-1}
For a geodesic triangle $\Delta_{u,v,w}$, 
let $\curvgrom(\Delta_{u,v,w})$ be 
the minimum number such that
$\overline{u,v}$ lies 
in a $\curvgrom(\Delta_{u,v,w})$-neighborhood
of $\overline{u,w}\,\cup\,\overline{v,w}$,
\emph{\IE}, for every node $x$ on $\overline{u,v}$, there exists a node $y$ on 
$\overline{u,w}$ or $\overline{v,w}$ such that $\dist_G(x,y)\leq\curvgrom(\Delta_{u,v,w})$.
Then
the graph $G$ has a Gromov-hyperbolic curvature (or
Gromov hyperbolicity) of $\curvgrom\eqdef\curvgrom(G)$ 
where 
$\curvgrom(G)=\min\limits_{u,v,w\in V} \left\{ \curvgrom(\Delta_{u,v,w}) \right\}$.
\end{definition}

An infinite collection $\mathcal{G}$ of graphs
belongs to the class of $\curvgrom$-Gromov-hyperbolic graphs 
if and only if 
any graph $G\in\mathcal{G}$ has a Gromov-hyperbolic curvature of $\curvgrom$.
Informally, any infinite metric space has a finite value of $\curvgrom$ 
if it behaves metrically in the large scale as a \emph{negatively curved} Riemannian manifold, and thus 
the value of $\curvgrom$ can be related to the other standard 
curvatures of a hyperbolic manifold. 
For example, a simply connected complete Riemannian manifold whose sectional curvature is below $\alpha<0$ has a value of 
$\curvgrom=O\big(\sqrt{-\alpha}\,\big)$ (see~\cite{R96}).
This is a major justification of using $\curvgrom$ as a notion of curvature of any metric space.

Let $\omega$ be the value such that two $n\times n$ matrices can be multiplied in $O(n^\omega)$ time; 
the smallest current value of $\omega$ is about {$2.373$~\cite{VVW12}}. Then the following results computational complexity
results are known for computing 
$\curvgrom(G)$
for an $n$-node graph $G$.
\begin{enumerate}[label=$\triangleright$,leftmargin=0.7cm]
\item
$\curvgrom(G)$
can be \emph{exactly} computed 
in
$O\left(n^{\frac{5+\omega}{2}}\right)=O\left(n^{3.687}\right)$ 
{time~\cite{ipl15}}.
\item
$(1+\eps)$-approximation of 
$\curvgrom(G)$
can be computed 
in
\\
$\tilde{O}\left(\frac{1}{\eps} n^{1+\omega}\right)=\tilde{O}\left(\frac{1}{\eps} n^{3.373}\right)$ 
{time~\cite{Du14}}, and 
$(2+\eps)$-approximation of 
\\
$\curvgrom(G)$
can be computed 
in
$\tilde{O}\left(\frac{1}{\eps} n^{\omega}\right)=\tilde{O}\left(\frac{1}{\eps} n^{2.373}\right)$ 
{time~\cite{Du14}}\footnote{$\tilde{O}(\cdot)$ is a standard computational complexity 
notation that omits poly-logarithmic factors.}.
\item
$8$-approximation of 
$\curvgrom(G)$
can be computed 
in
$O\left(n^2\right)$ 
time~\cite{CCDDMV18}.
\end{enumerate}
It is easy to see that if $G$ is a tree then 
$\curvgrom(G) = 0$. 
Other examples of graph classes for which 
$\curvgrom(G)$ 
is a small constant include
{\em chordal graphs}, {\em cactus of cliques}, 
{\em AT-free} graphs, {\em link graphs of simple polygons}, 
and {\em any} class of graphs with a {\em fixed} diameter.
A small value of Gromov-hyperbolicity is often crucial for algorithmic designs; for example,
several routing-related problems or the diameter estimation problem become easier for
networks with small $\curvgrom$ values~\cite{CE07,CDEHV08,CDEHVX12,GL05}.
There are many well-known measures of curvature of a continuous surface or other similar spaces (\EG,  
curvature of a manifold) that are widely used in many branches of physics and mathematics.
It is possible to relate Gromov-hyperbolic curvature to such other curvature notions indirectly via 
its scaled version, \EG, see~\cite{JLB07,NS11,JLA11}.

\subsubsection{Gromov-hyperbolic curvature and real-world networks}
\label{sec-real-grom}

Recently, there has been a surge of empirical works measuring and analyzing the Gromov curvature $\curvgrom$ 
of networks, and many real-world networks 
(\EG, preferential attachment networks, networks of high power transceivers in 
a wireless sensor network,
communication networks at the IP layer and at other levels)
were observed to have a small constant value of $\curvgrom$~\cite{NS11,PKBV10,JL04,JLB07,ALJKZ08}.
The authors in~\cite{ADM14} 
analyzed $11$ well-known biological networks 
and $9$ well-known social networks
for their  $\curvgrom$ values and found all but one network had a \emph{statistically significant small} value
of $\curvgrom$.
These references also describe implications of range of $\curvgrom$
on the actual real-world applications of these networks.
As mentioned in the following subsection, the 
Gromov-hyperbolicity measure
is \emph{fundamentally different} from the commonly used topological properties for a graph; 
for example, it is \emph{neither} a hereditary \emph{nor} a monotone property, 
is \emph{not} the same as tree-width measure or 
other standard combinatorial properties that are commonly used in the computer science literature, and 
\emph{not} necessarily a measure of closeness to tree topology.

\subsubsection{Some clarifying remarks regarding Gromov-hyperbolicity measure}

As pointed out in details by the authors in~\cite[Section 1.2.1]{DKMY18}, the 
Gromov-hyperbolicity measure 
$\curvgrom$
enjoys many non-trivial topological characteristics.
In particular, the authors in~\cite[Section 1.2.1]{DKMY18} point out the following:
\begin{enumerate}[label=$\triangleright$,leftmargin=0.7cm]
\item
$\curvgrom$
is \emph{not} a hereditary or monotone property
since removal of nodes or edges may change the value of $\curvgrom$ sharply. 
\item
$\curvgrom$
is \emph{not} necessarily the same as tree-width measure (see also~\cite{MSV11,ADM14}), or
other standard combinatorial properties (\EG, betweenness centrality, clustering coefficient,
dense sub-graphs) that are commonly used in the computer science literature.
\item
``Close to hyperbolic topology'' is \emph{not} necessarily the same as ``close to tree topology''.
\end{enumerate}

\subsection{Geometric curvatures}
\label{sec-geom-defn}

In this section, we describe geometric curvatures of graphs 
by using correspondence with topological objects in higher dimension.
The approach of using associations of sub-graphs with 
with topological objects in higher dimension has also been used in some 
previous papers such as~\cite{WJS16} but our anomaly detection approach is quite 
different from them.

\subsubsection{Basic topological concepts}
\label{sec-app-topo}

We first review some basic concepts from topology; see introductory 
textbooks such as~\cite{H94,GG99} for further information.
Although not absolutely necessary, the reader may find it useful to think of 
the underlying metric space 
as the $r$-dimensional real space $\R^r$ be for some integer $r>1$.
\begin{enumerate}[label=$\blacktriangleright$,leftmargin=0.7cm]
\item
A subset $S\subseteq \R^r$ is \emph{convex} if and only if for any $x,y\in S$, 
the \emph{convex combination} of $x$ and $y$ 
is also in $S$.
\item
A set of $k+1$ points $x_0,\dots,x_k\in \R^r$ are called \emph{affinely independent} if and only if 
for all $\alpha_0,\dots,\alpha_k\in\R$ 
$\sum_{j=0}^k \alpha_j x_j=0$ 
and 
$\sum_{j=0}^k \alpha_j=0$ 
implies $\alpha_0=\dots=\alpha_k=0$.
\item
The $k$-\emph{simplex} generated by a set of $k+1$ affinely independent points 
$x_0,\dots,x_k\in \R^r$ is the subset 
$\cS\big(x_0,\dots,x_k\big)$
of $\R^r$
generated by \emph{all} convex combinations of 
$x_0,\dots,x_k$. 
\begin{enumerate}[label=$\triangleright$,leftmargin=0.7cm]
\item
Each $(\ell+1)$-subset 
$
\big\{x_{i_0},\dots,x_{i_\ell}\big\}
\subseteq
\big\{x_0,\dots,x_k\big\}
$
defines the $\ell$-simplex 
$\cS\big(x_{i_0},\dots,x_{i_\ell}\big)$
that is called a \emph{face} of dimension $\ell$ (or a $\ell$-\emph{face}) of 
$\cS\big(x_0,\dots,x_k\big)$.
A $(k-1)$-face, $1$-face and $0$-face is called a \emph{facet}, an \emph{edge} and a \emph{node}, respectively.
\end{enumerate}
\item
A (closed) \emph{halfspace} is a set of points satisfying 
$\sum_{j=1}^r a_j x_j \leq b$ for some $a_1,\dots,a_r,b\in\R$. 
The convex set obtained by a bounded non-empty intersection of a finite number of halfspaces is called 
a \emph{convex polytope}
(\emph{convex polygon} in two dimensions). 
\begin{enumerate}[label=$\triangleright$]
\item
If the intersection of a halfspace and a convex polytope is a subset of the halfspace then 
it is called a \emph{face} of the polytope. Of particular interests are faces of 
dimensions $r-1$, $1$ and $0$, which are called \emph{facets}, \emph{edges} and \emph{nodes} of the polytope, respectively.
\end{enumerate}
\item
A \emph{simplicial complex} (or just a complex) 
is a topological space constructed by the union of simplexes via topological associations.
\end{enumerate}

\subsubsection{Geometric curvature definitions}
\label{sec-gmeq1}

Informally, a \emph{complex} is ``glued'' from nodes, edges and polygons via topological identification.
We first define $k$-complex-based \emph{Forman's combinatorial Ricci curvature}
for elementary components (such as nodes, edges, triangles and higher-order cliques)  
as described in~\cite{Bl14,Fo03,WJS16,WSJ16}, and then obtain a scalar curvature that takes an appropriate 
linear combination of these values (via Gauss-Bonnet type theorems, see for example~\cite[Sections $4.1$--$4.3$]{WSJ16}  
and the references therein)
that correspond to the so-called 
\emph{Euler characteristic} of the complex 
that is 
topologically associated with the given graph.  
In this paper, we consider such Euler characteristics of a graph to define geometric curvature.

To begin the topological association, we (topologically) associate a $q$-simplex with a $(q+1)$-clique $\cK_{q+1}$; 
for example, $0$-simplexes, $1$-simplexes, $2$-simplexes and $3$-simplexes
are associated with nodes, edges, $3$-cycles (triangles) and $4$-cliques, respectively.
Next, we would also need the concept of an ``order'' of a simplex for more
non-trivial topological association. 
Consider a $p$-face $f^p$ of a $q$-simplex.
An order $d$ association of such a face, which we will denote by the notation $f_d^p$ 
with the additional subscript $d$,
is associated with a sub-graph of \emph{at most} $d$ nodes that is obtained by starting with 
$\cK_{p+1}$ and then \emph{optionally} replacing each edge by a path between the two nodes. 
For example, 
\begin{itemize}
\item 
$f_d^0$ is a node of $G$ for all $d\geq 1$.
\item 
$f_2^1$ is an edge, and 
$f_d^1$ for $d>2$ is a path having at most $d$ nodes between two nodes adjacent in $G$.
\item 
$f_3^2$ is a triangle (cycle of $3$ nodes or a $3$-cycle), and 
$f_d^2$ for $d>3$ is obtained from $3$ nodes by connecting every pair of nodes by a path such that 
the total number of nodes in the sub-graph is at most $d$.
\end{itemize}
%
Naturally, the higher the values of $p$ and $q$ are, the more complex are the topological associations.
Let $\cF_d^k$ be the set of all $f_d^k$'s that are topologically associated. 
With such associations via $p$-faces of order $d$, the
Euler characteristics of the graph $G=(V,E)$ and consequently the curvature can be defined as
\begin{gather}
\mathfrak{C}^p_d(G)\eqdef
\sum\limits_{k=0}^p (-1)^k\, \left| \cF_d^k \right|
\label{gmeq1}
\end{gather}
It is easy to see that both 
$\mathfrak{C}^0_d(G)$ 
and 
$\mathfrak{C}^1_d(G)$ 
are too simplistic to be of use in practice. Thus, we consider the next higher value
of $p$ in this paper, namely when $p=2$. 
Letting $\cC(G)$ denote the number of cycles of at most $d+1$ nodes in $G$, 
we get the measure 
\begin{gather*}
\mathfrak{C}^2_d(G) = |V|-|E|+ |\cC(G)|
\end{gather*}

\noindent
\textbf{Suitability of geometric curvature measures for real-world networks}: 
The usefulness of geometric curvatures for real-world networks
was demonstrated in publications such as~\cite{WJS16,WSJ16,SSG18}.

\section{Formalizations of two anomaly detection problems on networks}
\label{def-prob}

In this section, we formalize two versions of the anomaly detection problem on networks. 
An underlying assumption on the behind these formulations is that 
the graph adds/deletes \emph{edges} only while keeping the same set of nodes.

\subsection{Extremal anomaly detection for static networks}
\label{sec-static}

The problems in this subsection are motivated by a desire to quantify the extremal sensitivity 
of static networks. The basic decision 
question is: ``\emph{is there a subset among a set of prescribed edges whose deletion may change the network
curvature significantly}?''. This directly leads us to the following \emph{decision} problem:

\smallskip

\begin{flushleft}
\begin{tabular}{r l }
\toprule
$\!\!\!\!\!\!$
\textbf{Problem name}: & Extremal Anomaly Detection Problem 
\\
                & (\eadp$_{\mathfrak{C}}(G,\wte,\gamma)$)
\\
\textbf{Input}: & $\bullet$ A curvature measure $\mathfrak{C}:G\mapsto\R$ 
\\
                & $\bullet$ A connected graph $G=(V,E)$, 
\\
								& \hspace*{0.1in} an edge subset $\wte\subseteq E$ such that $G\setminus\wte$ is connected,
\\
                & \hspace*{0.1in} and a real number $\gamma<\mathfrak{C}(G)$ (\emph{resp}., $\gamma>\mathfrak{C}(G)$)
\\
\begin{tabular}{r}
\textbf{Decision}
\\
\textbf{question}
\end{tabular}
$\!\!\!\!\!$: & is there an edge subset $\whe\subseteq\wte$ such that 
                              $\mathfrak{C}(G\setminus \whe)\leq\gamma$
\\              
                & \hspace*{1.85in} (\emph{resp}., $\mathfrak{C}(G\setminus \whe)\geq\gamma$) ?
\\
\begin{tabular}{r}
\textbf{Optimization}
\\
\textbf{question}
\end{tabular}
$\!\!\!\!\!$: & 
							 $\!\!\!\!\!$
               \begin{tabular}{l}
                if the answer to the decision question is ``yes'' 
								\\
								\hspace*{0.3in} then minimize $|\whe|$
               \end{tabular}
\\
\textbf{Notation}: & if the answer to the decision question is ``yes'' then 
\\
                   & \hspace*{0.2in} 
									   the minimum possible value of $|\whe|$ 
\\
                   & \hspace*{0.2in} 
										 is denoted by 
										 $\mathsf{OPT}_{\mbox{\eadp}_{\mathfrak{C}}}(G,\wte,\gamma)$
\\
\bottomrule
\end{tabular}
\end{flushleft}

\smallskip

\noindent
The following comments regarding the above formulation should be noted:
\begin{enumerate}[label=$\triangleright$,leftmargin=0.7cm]
\item
For the case $\gamma<\mathfrak{C}(G)$ (\emph{resp}., $\gamma>\mathfrak{C}(G)$) 
we allow 
$\mathfrak{C}(G\setminus \wte)>\gamma$
(\emph{resp}., $\mathfrak{C}(G\setminus \wte)<\gamma$), thus $\whe=\wte$ 
need \emph{not} be a feasible solution at all.
\item
The curvature function is only defined for connected graphs, thus we require
$G\setminus\wte$ to be connected.
\item
The edges in $E\setminus\wte$ can be thought of as ``critical'' edges 
needed for the functionality of the network. For example, in the context of 
inference of minimal biological networks from indirect experimental evidences~\cite{ADDKSZW07,ADDS},
the set of critical edges represent direct biochemical interactions with concrete evidence.
\end{enumerate}

\subsection{Targeted anomaly detection for dynamic networks}
\label{sec-dynamic}

These problems are primarily motivated by change-point detections between two successive discrete time steps 
in dynamic networks~\cite{AC17,KS09}, but they can also be applied to static networks when a subset of the 
final desired network \emph{is} known.
\FI{ex1-fig} illustrates targeted anomaly detection for a dynamic biological network.

\bigskip
\begin{flushleft}
\begin{tabular}{r l}
\toprule
$\!\!\!\!\!\!$
\begin{tabular}{r}
\textbf{Problem}
\\
\textbf{name}
\end{tabular}
$\!\!\!\!\!$: & Targeted Anomaly Detection Problem (\badp$_{\mathfrak{C}}(G_1,G_2)$)
\\
\textbf{Input}: & $\bullet$ Two connected graphs $G_1=(V,E_1)$ and $G_2=(V,E_2)$ 
\\
                & \hspace*{0.2in} with $E_2\subset E_1$
\\
                & $\bullet$ A curvature measure $\mathfrak{C}:G\mapsto\R$ 
\\
\begin{tabular}{r}
\textbf{Valid}
\\
\textbf{solution}
\end{tabular}
$\!\!\!\!\!$: & 
an edge subset $E_3\subseteq E_1\setminus E_2$ such that 
$\mathfrak{C}(G_1\setminus E_3)=\mathfrak{C}(G_2)$. 
\\
\textbf{Objective}: & \emph{minimize} $|E_3|$.
\\
\textbf{Notation}: & 
							        $\!\!\!\!\!$
                      \begin{tabular}{l}
                         the minimum value of $|E_3|$ 
										     \\
                         \hspace*{0.2in} 
                         is denoted by $\mathsf{OPT}_{\mbox{\badp}_{\mathfrak{C}}}(G_1,G_2)$
                      \end{tabular}
\\
\bottomrule
\end{tabular}
\end{flushleft}

\section{Computational complexity of extremal anomaly detection problems}
\label{sec-proof-extreme}

\subsection{Geometric curvatures: computational complexity of \eadp$_{\mathfrak{C}^2_d}$} 
\label{sec-proof-extreme-geometric}

\begin{theorem}\label{ext-thm}~\\
\textbf{\emph{(\emph{a})}}
The following statements hold for \eadp$_{\mathfrak{C}^2_d}(G,\wte,\gamma)$
when $\gamma>\mathfrak{C}^2_d(G)$:
\medskip
\begin{adjustwidth}{0.1in}{}
\begin{description}
\item[(\emph{a}1)]
We can decide in polynomial time the answer to the 
decision question (\emph{\IE}, if there exists any feasible solution $\whe$ or not).
\item[(\emph{a}2)]
If a feasible solution exists then the following results hold: 
\begin{description}
\item[(\emph{a}2-1)]
Computing 
$\mathsf{OPT}_{\mbox{\eadp}_{\mathfrak{C}^2_d}}(G,\wte,\gamma)$ 
is $\NP$-hard for 
all $d$ that are multiple of $3$.
\item[(\emph{a}2-2)]
If $\gamma$ is sufficient larger than $\mathfrak{C}^2_d(G)$ 
then we can design an approximation algorithm that approximates both 
the cardinality of the minimal set of edges for deletion and the absolute difference between the
two curvature values. More precisely, 
if $\gamma \geq \mathfrak{C}^2_d(G) + \left(\frac{1}{2}+\eps\right)(2|\wte|-|E|)$ 
for some $\eps>0$,
then we can find in polynomial time
a subset of edges $E_1\subseteq\wte$ such that 
\[
|E_1|
\leq 2 \, \mathsf{OPT}_{\mbox{\eadp}_{\mathfrak{C}^2_d}}(G,\wte,\gamma)
\text{ and }
\frac {\mathfrak{C}^2_d(G\setminus E_1) - \mathfrak{C}^2_d(G) } 
{
\gamma-\mathfrak{C}^2_d(G)
}
\geq
\frac{4\eps}{1+2\eps} 
\]
\end{description}
\end{description}
\end{adjustwidth}
%
\textbf{\emph{(\emph{b})}}
The following statements hold for \eadp$_{\mathfrak{C}^2_d}(G,\wte,\gamma)$
when $\gamma<\mathfrak{C}^2_d(G)$:
\medskip
\begin{adjustwidth}{0.1in}{}
\begin{description}
\item[(\emph{b}1)]
We can decide in polynomial time the answer to the 
decision question (\emph{\IE}, if there exists any feasible solution $\whe$ or not).
\item[(\emph{b}2)]
If a feasible solution exists 
and $\gamma$ is not too far below $\mathfrak{C}^2_d(G)$ 
then we can design an approximation algorithm that approximates both 
the cardinality of the minimal set of edges for deletion and the absolute difference between the
two curvature values. More precisely, letting $\Delta$ denote the number of 
cycles of $G$ of at most $d+1$ nodes that contain at least one edge from $\wte$,
if $\gamma \geq \mathfrak{C}^2_d(G) - \frac{\Delta}{1+\eps}$ 
for some $\eps>0$
then we can find in polynomial time
a subset of edges $E_1\subseteq\wte$ such that 
\[
|E_1|
\leq 2 \, \mathsf{OPT}_{\mbox{\eadp}_{\mathfrak{C}^2_d}}(G,\wte,\gamma)
\text{ and }
\frac {\mathfrak{C}^2_d(G\setminus E_1) - \mathfrak{C}^2_d(G) } 
{
\gamma-\mathfrak{C}^2_d(G)
}
\leq
1 - \eps
\]
%
\item[(\emph{b}3)]
If $\gamma<\mathfrak{C}^2_d(G)$ then, even if 
$\gamma=\mathfrak{C}^2_d(G\setminus\wte)$ (\emph{\IE}, a trivial feasible solution exists), 
computing 
$\mathsf{OPT}_{\mbox{\eadp}_{\mathfrak{C}^2_d}}(G,\wte,\gamma)$
is at least as hard as computing 
\badp$_{\mathfrak{C}^2_d}(G_1,G_2)$ and therefore \textbf{all} the hardness results 
for \badp$_{\mathfrak{C}^2_d}(G_1,G_2)$
in Theorem~\ref{thm-badp} also apply to 
$\mathsf{OPT}_{\mbox{\eadp}_{\mathfrak{C}^2_d}}(G,\wte,\gamma)$.
\end{description}
\end{adjustwidth}
\end{theorem}

\subsubsection{Proof techniques and relevant comments regarding Theorem~\ref{ext-thm}}
\label{informal-ext-thm}

\paragraph{\em\bf(on proofs of \textbf{(\emph{a}1)} and \textbf{(\emph{b}1)})}
After eliminating a few ``easy-to-solve'' sub-cases, 
we prove the remaining cases of \textbf{(\emph{a}1)} and \textbf{(\emph{b}1)}
by reducing the feasibility questions to suitable minimum-cut problems; the reductions and proofs are somewhat
different due to the nature of the objective function. It would of course be of interest if a single algorithm and proof
can be found that covers both instances and, more importantly, if a direct and more efficient 
greedy algorithm can be found that \emph{avoids} the maximum flow computation. 

\paragraph{\em\bf(on proofs of \textbf{(\emph{a}2-2)} and \textbf{(\emph{b}2)})}
Our general approach to 
prove \textbf{(\emph{a}2-2)} and \textbf{(\emph{b}2)} is to formulate 
these problems as a series of (provably $\NP$-hard and polynomially many) ``constrained'' minimum-cut 
problems.
We start out with \emph{two different} (but well-known) polytopes for the 
minimum cut problem (polytopes~\eqref{pto1} and \eqref{pto1}$'$). 
Even though the polytope~\eqref{pto1}$'$ is of exponential size for general graphs, 
it is of polynomial size for our particular minimum cut version and so we do \emph{not} need to 
appeal to separation oracles for its efficient solution.
We subsequently add extra constraints corresponding to a parameterized version of the minimization 
objective and solve the resulting augmented polytopes (polytopes~\eqref{pto2} and \eqref{pto2}$'$)
in polynomial time to get a fractional solution and use a simple deterministic rounding scheme to 
obtain the desired bounds. 
\begin{enumerate}[label=$\triangleright$,leftmargin=0.7cm]
\item
Our algorithmic approach uses a sequence of $\lceil \log_2 (1 + |\wte|) \rceil=O(\log |E|)$ linear-programming ($\LP$)
computations by using an obvious binary search over the relevant parameter range.
It would be interesting to see if we can do the same using $O(1)$ $\LP$ computations. 
\item
Is the factor $2$ in 
``$|E_1| \leq 2 \, \mathsf{OPT}_{\mbox{\eadp}_{\mathfrak{C}^2_d}}(G,\wte,\gamma)$''
an artifact of our specific rounding scheme around the threshold of $\nicefrac{1}{2}$ and perhaps
can be improved using a cleverer rounding scheme? 
This seems \emph{unlikely} for the case when 
$\gamma<\mathfrak{C}^2_d(G)$ since the inapproximability results in 
\textbf{(\emph{b}3)}
include a 
$(2-\eps)$-inapproximability 
assuming the unique games conjecture is true.
However, this possibility cannot be ruled out for the case when 
$\gamma>\mathfrak{C}^2_d(G)$ since we can only prove $\NP$-hardness for this case.
\item
There are subtle but crucial differences between the rounding schemes for 
\textbf{(\emph{a}2-2)} and \textbf{(\emph{b}2)} that is essential to proving the desired bounds.
To illustrate this, consider an edge $e$ with a fractional value of $\nicefrac{1}{2}$ for its corresponding variable.
In the rounding scheme~\eqref{rs1} of \textbf{(\emph{a}2-2)}
$e$ will only \emph{sometimes} be designated as a cut edge,  whereas 
in the rounding scheme~\eqref{rs1}$'$ of \textbf{(\emph{b}2)}
$e$ will \emph{always} be designated as a cut edge.
\end{enumerate}

\paragraph{\em\bf (on the bounds over $\pmb{\gamma}$ in (\emph{a}2-2)} 
If $|\wte| \lessapprox \frac{1}{2}|E|$ then the condition on $\gamma$ \emph{is} redundant (\IE, \emph{always} holds).
Thus indeed the $2$-approximation is likely to hold \emph{unconditionally} for \emph{practical} applications of this problem since 
anomaly is supposed to be caused by a \emph{large} change in curvature by a relatively 
\emph{small} number of elementary components (edges in our cases).

Furthermore, 
if $|E|\leq 2|V|$ then the condition on $\gamma$ always holds \emph{irrespective} of the value of $|\wte|$, 
and the \emph{smaller} is $|\wte|$ with respect to $|E|$ the \emph{better} is our approximation of the curvature 
difference. 
As a general illustration, 
when 
$\eps=\nicefrac{1}{5}$
the assumptions are 
$
\gamma \geq \mathfrak{C}^2_d(G) + \frac{7}{10}(2|\wte|-|E|)
$,
and the corresponding bounds are
$
|E_1| \leq 2 \, \mathsf{OPT}_{\mbox{\eadp}_{\mathfrak{C}^2_d}}(G,\wte,\gamma) 
\,\,\,\,
\text{and}
\,\,\,\,
\frac { \mathfrak{C}^2_d(G\setminus E_1) - \mathfrak{C}^2_d(G) }
{ \gamma-\mathfrak{C}^2_d(G) }
\geq
\frac{4}{7} 
$.

\paragraph{\em\bf(on the hardness proof in \textbf{(\emph{a}2-1)})}
Our reduction is from the \emph{densest-$k$-subgraph} (\dkst) problem. 
We use the reduction from the CLIQUE problem to \dkst 
detailed by Feige and Seltser in~\cite{FS97} which shows that 
\dkst is $\NP$-hard even if the degree of every node is at most $3$. 
For convenience in doing calculations, 
we use the reduction of Feige and Seltser
starting from the still $\NP$-hard version of the CLIQUE problem where the input instances are
$(n-4)$-regular $n$-node graphs. 
Pictorially, the reduction is illustrated in \FI{fig3}.
Note that \dkst is \emph{not} known to be $(1+\eps)$-inapproximable assuming P$\neq\NP$ (though it is
likely to be), and thus our particular reduction cannot be generalized to
$(1+\eps)$-inapproximability assuming P$\neq\NP$.

\begin{figure}[htbp]
\centerline{\includegraphics[scale=0.6]{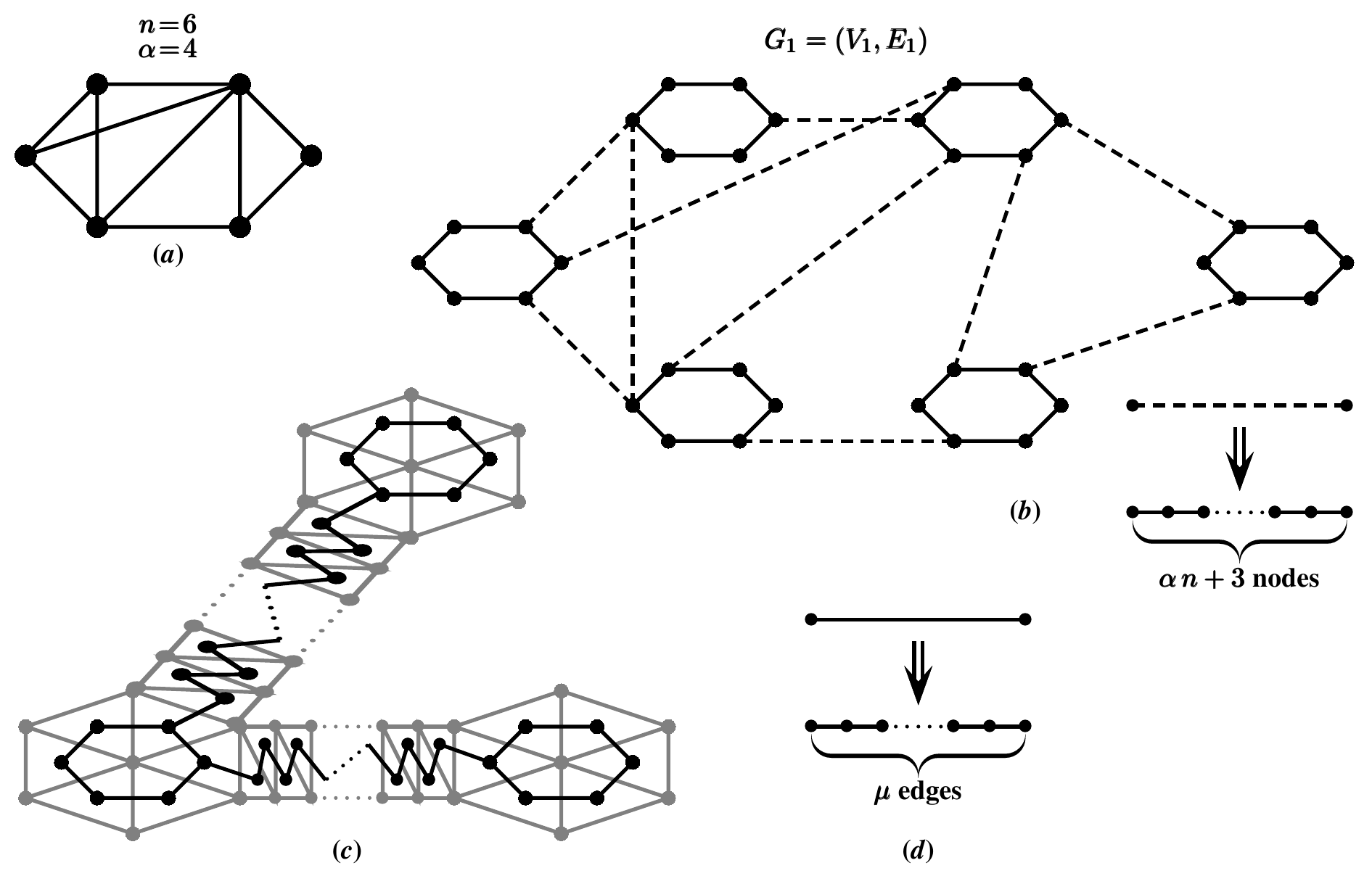}}
\caption{\label{fig3}Illustration of the reduction in the proof of Theorem~\ref{ext-thm}(\emph{a}2-1).
\textbf{(\emph{a})} 
The original instance of the $\alpha$-CLIQUE problem on an $n$-node graph.
\textbf{(\emph{b})} 
Illustration of the $\NP$-hardness reduction 
from $\alpha$-CLIQUE to \dkst by Feige and Seltser~\cite{FS97}.
\textbf{(\emph{c})} 
Illustration of associations of nodes with unique $3$-cycles such that an edge between two nodes are adjacent 
correspond to sharing an unique edge of 
their associated $3$-cycles. 
\textbf{(\emph{d})} 
Splitting of every edge of $G_1$ into a path of $\mu$ edges 
for the case when $d=3\mu$ for some integer $\mu>1$.
}
\end{figure}

\subsubsection{Proof of Theorem~\ref{ext-thm}}
\label{sec-app-proof-ext-thm}

\noindent
\textbf{Proof of (\emph{a}1)}

Let the notation $\cC(H)$ denote the set of cycles having at most $d+1$ nodes in a graph $H$. 
Assume $\Delta=|\cC(G)|$ and let $\cC(G)=\set{\cF_1,\cF_2,\dots,\cF_\Delta}$; 
thus $\mathfrak{C}^2_d(G)=n-m+\Delta$ where $|V|=n$ and $|E|=m$.
Since $d$ is fixed, $\Delta=O(n^d)$ and all the cycles in $\cC(G)$ 
can be explicitly enumerated in polynomial ($O(n^d)$) time. 
Let $\cC'(G)=\set{\cF_1,\cF_2,\dots,\cF_{\Delta'}}\subseteq\cC(G)$
be the set of $\Delta'\leq\Delta$ 
cycles in $\cC(G)$ that involve one of more edges from $\wte$.
An \emph{overview} of the main steps in our proof for (\emph{a}1) is as follows.
\begin{align*}
\text{\bf 1.} &
  \text{We identify sub-cases that are easy to solve.}
\\
\text{\bf 2.} &
  \text{For all remaining sub-cases, we reduce our problem to a standard}
\\
                   &
  \text{(directed) minimum $s$-$t$ cut problem such that the following}
\\
                   &
  \text{statements hold:}
\\
                   &
  \hspace*{0.2in}
  \text{$\triangleright$ The cut network can be constructed in polynomial time.}
\\
                   &
  \hspace*{0.2in}
  \text{$\triangleright$ There exists a feasible solution of \eadp$_{\mathfrak{C}^2_d}(G,\wte,\gamma)$ 
	        if and}
\\
                   &
  \hspace*{0.35in} 
	\text{only if the minimum cut value is at most $\Gamma\eqdef\mathfrak{C}^2_d(G)-\gamma + |\wte|$.
	}
\end{align*}

\noindent
\textbf{Step 1. Identifying sub-cases that are easy to solve}

\medskip
We first observe that the following sub-cases are easy to solve:
\begin{itemize}
\item
If $\gamma>n-(m-|\wte|)+\Delta$ 
then we can assert that there is no feasible solution. 
This is true because 
for any 
$E'\subseteq \wte$ it is true that
$\mathfrak{C}^2_d(G\setminus E')$ is 
at most $n-(m-|\wte|)+\Delta$.
\item
If 
$\gamma\leq n-(m-|\wte|)+\Delta$ 
and 
$\Delta'=0$ then 
there exists a \emph{trivial} optimal feasible solution of the following form:  
\begin{quote}
select any set of $m_1$ edges from $\wte$ 
where 
$m_1$ is the least positive integer satisfying 
$n-(m_1-|\wte|)+\Delta\geq\gamma$.  
\end{quote}
\end{itemize}

\noindent
\textbf{Step 2. Solving all remaining sub-cases}

\medskip
We assume that 
$\gamma\leq n-(m-|\wte|)+\Delta$ 
and $\Delta'>0$.
Consider a subset $E_1\subseteq \wte$ of $m_1=|E_1|\leq |\wte|$ edges for deletion and suppose that removal 
of the edges in $E_1$ removes $\Delta_1\leq\Delta'$ cycles from $\cC'(G)$ (\IE, $|\cC'(G\setminus E_1)|=\Delta'-\Delta_1$).
Then,
\begin{multline}
\mathfrak{C}^2_d(G\setminus E_1)
=
n-(m-m_1)+(\Delta-\Delta_1)
\\
=
n-m+\Delta+(m_1-\Delta_1)
=
\mathfrak{C}^2_d(G)+(m_1-\Delta_1)
\label{eq-minus-e1}
\end{multline}
and consequently one can observe that 
\begin{multline}
\mathfrak{C}^2_d(G\setminus E_1) \geq \gamma
\,\equiv\,
m_1-\Delta_1 \geq  \gamma-\mathfrak{C}^2_d(G)
\,\equiv\,
\Delta_1-m_1 \leq  \mathfrak{C}^2_d(G) - \gamma
\\
\equiv\,
\Delta_1 + (m-m_1)  \leq   \mathfrak{C}^2_d(G)-\gamma + m
\,\equiv\,
\Delta_1 + (|\wte|-m_1)  \leq   \mathfrak{C}^2_d(G)-\gamma + |\wte|
	 \eqdef \Gamma
\label{eq-mc}
\end{multline}
Note that 
$
\Gamma 
=
\mathfrak{C}^2_d(G)-\gamma + |\wte|
=
n-(m-|\wte|)+\Delta-\gamma
\geq 0
$ 
and $|\wte|-m_1$ is the number of edges in $\wte$ that 
are \emph{not} in $E_1$ and therefore not selected for deletion.
Also, note that 
$\Gamma$ is a quantity that depends on the problem instance only and 
\emph{does not change if one or more edges are deleted}.
Based on this interpretation, we
construct the following instance (digraph) 
$\cG=(\cV,\cE)$ of a (standard directed) minimum $s$-$t$ cut problem (where $\ccap(u,v)$ 
is the capacity of a directed edge $(u,v)$): 
\begin{itemize}
\item
The nodes in $\cV$ are as follows: 
a source node $s$, a sink node $t$, a node (an ``edge-node'') $u_e$ 
for every edge $e\in \wte$ and a node (a ``cycle-node'') $u_{\cF_i}$ for every cycle $\cF_i\in\cC'(G)$.
The total number of nodes is therefore $O(|\wte|+n^d)$, \IE, polynomial in $n$.
\item
The directed edges in $\cE$ and their corresponding capacities are as follows: 
\begin{itemize}
\item
For every edge $e\in \wte$, 
we have a directed edge $(s,u_e)$ (an ``edge-arc'') of capacity 
$\ccap(s,u_e)=1$.
\item
For every cycle $\cF_i\in\cC'(G)$, 
we have a directed edge (a ``cycle-arc'') $(u_{\cF_i},t)$ of capacity 
$\ccap(u_{\cF_i},t)=1$.
\item
For every cycle $\cF_i\in\cC'(G)$ and 
every edge $e\in \wte$ such that $e$ is an edge of $\cF_i$, we have a directed edge
(an ``ed-cy-arc'', ed-cy-arc for short) $(u_e,u_{\cF_i})$ of capacity $\ccap(u_e,u_{\cF_i})=\infty$.
\end{itemize}
\end{itemize}
For an $s$-$t$ cut $(\cS,\cV\setminus\cS)$ of $\cG$ (where $s\in\cS$ and $t\notin\cS$),
let 
$\cut(\cS,\cV\setminus\cS)=\set{(x,y) \,|\, x\in\cS,\,y\notin\cS}$
and 
$\ccap(\cut(\cS,\cV\setminus\cS))=\sum_{(x,y)\in\cut(\cS,\cV\setminus\cS)} \ccap(x,y)$
denote the edges in the cut and the capacity of the cut, respectively.
It is well-known how to compute a \emph{minimum} $s$-$t$ cut of value 
\\
$
\Phi \eqdef \min_{\emptyset\subset\cS\subset\cV,\, s\in\cS,\, t\notin\cS} \{\ccap(\cut(\cS,\cV\setminus\cS)) \}
$
in polynomial time~\cite{CCPS97}.
The following lemma proves part 
(\emph{a}1) of the theorem.

\begin{lemma}
There exists any feasible solution of 
\eadp$_{\mathfrak{C}^2_d}(G,\wte,\gamma)$
if and only if $\Phi\leq \Gamma$.
Moreover, if 
$(\cS,\cV\setminus\cS)$ 
is a minimum $s$-$t$ cut 
of $\cG$ of value 
$\Phi\leq\Gamma$ 
then $\whe=\set{e \,|\, u_e\in\cS}$
is a feasible solution for
\eadp$_{\mathfrak{C}^2_d}(G,\wte,\gamma)$.
\end{lemma}

\begin{proof}
Suppose that there exists a feasible solution 
$E_1\subseteq\wte$ with $m_1=|E_1|$ edges 
for 
\eadp$_{\mathfrak{C}^2_d}(G,\wte,\gamma)$, and 
suppose that removal 
of the edges in $E_1$ removes $\Delta_1$ cycles from $\cC'(G)$.
Consider the cut $(\cS,\cV\setminus\cS)$ where 
\[
\cS = \{s\} \bigcup \{ u_e \,|\, e\in E_1 \} 
\bigcup 
\left\{
u_{\cF_i} \,|\, \text{$\cF_i$ contains at least one edge from $E_1$}
\right\}
\]
Note that \emph{no} ed-cy-arc belongs to 
$\cut(\cS,\cV\setminus\cS)$
and therefore 
\begin{multline*}
\ccap(\cut(\cS,\cV\setminus\cS))
\\
=
| \set{(s,u_e) \,|\, e\notin E_1} | + 
| \set{(u_{\cF_i},t) \,|\, \text{$\cF_i$ contains at least one edge from $E_1$} } |
\\
=
(|\wte|-m_1)+\Delta_1
\end{multline*}
and thus by Inequality~\eqref{eq-mc} we can conclude that 
\begin{multline*}
\mathfrak{C}^2_d(G\setminus E_1) \geq \gamma
\\
\equiv\,
\Delta_1 + (|\wte|-m_1)  \leq \Gamma
\,\equiv\,
\ccap(\cut(\cS,\cV\setminus\cS)) \leq\Gamma
\,\Rightarrow 
\Phi \leq \ccap(\cut(\cS,\cV\setminus\cS)) \leq\Gamma
\end{multline*}
For the other direction, consider a minimum $s$-$t$ cut 
$(\cS,\cV\setminus\cS)$ of $\cG$ of value $\Phi\leq\Gamma$.
Consider the solution 
$E_1 = \{ e \,|\, u_e\in\cS \} \subseteq\wte$
for 
\eadp$_{\mathfrak{C}^2_d}(G,\wte,\gamma)$, 
and suppose that removal 
of the edges in $E_1$ removes $\Delta_1$ cycles from $\cC'(G)$.
Since $\cG$ admits a trivial $s$-$t$ cut $(\{s\},\cV\setminus\{s\})$ of capacity $m_1<\infty$, 
\emph{no} ed-cy-arc can be an edge of \emph{any} minimum $s$-$t$ cut of $\cG$, \IE, 
$\cut(\cS,\cV\setminus\cS)$ contains only edge-arcs or cycle-arcs.
Let $E_2 = \{ \cF_j \,|\, u_{\cF_j}\in\cS \}$.
Consider an edge $e\in E_1$ and let 
$\cF_j$ be a cycle in $\cC'(G)$ containing $e$. Since 
$\cut(\cS,\cV\setminus\cS)$ contains 
no ed-cy-arc, it does \emph{not} contain the arc 
$(u_e,u_{\cF_j})$. It thus follows that 
the cycle-node $u_{\cF_j}$ must also belong to $\cS$ and thus 
$|E_2| = \Delta_1$. Now note that 
\begin{multline*}
\Phi 
=
| \set{ u_e \,|\, u_e\notin\cS  } |
+
| \set{ u_{\cF_j} \,|\, u_{\cF_j} \in \cS  } |
\\
= (|\wte| - |E_1|) + \Delta_1 
\leq 
\Gamma = \mathfrak{C}^2_d(G)-\gamma + |\wte|
\\
\equiv \,
\mathfrak{C}^2_d(G\setminus E_1) = \mathfrak{C}^2_d(G) + |E_1| - \Delta_1 \geq \gamma 
\end{multline*}
\end{proof}

\noindent
This completes a proof for (\emph{a}1).

\bigskip
\noindent
\textbf{Proof of (\emph{a}2-2)}
We will reuse the proof of (\emph{a}1) as appropriate.
Let 
$\whe\subseteq\wte$ be an optimal solution of the optimization version of 
\eadp$_{\mathfrak{C}^2_d}(G,\wte,\gamma)$ having 
$\mathsf{OPT}_{\mbox{\eadp}_{\mathfrak{C}^2_d}}(G,\wte,\gamma)$ nodes.
Note that $\mathsf{OPT}_{\mbox{\eadp}_{\mathfrak{C}^2_d}}(G,\wte,\gamma)\in\set{1,2,\dots,|\wte|}$ and thus in polynomial time 
we can ``guess'' every possible value of 
\\
$\mathsf{OPT}_{\mbox{\eadp}_{\mathfrak{C}^2_d}}(G,\wte,\gamma)$, 
solve the corresponding optimization problem 
with this additional constraint, and take the best of these solutions. 
In other words, it suffices 
if we can find, under the assumption that 
$\mathsf{OPT}_{\mbox{\eadp}_{\mathfrak{C}^2_d}}(G,\wte,\gamma)=\kappa$ for some $\kappa\in\set{1,2,\dots,|\wte|}$,  
find a solution $E_1\subseteq\wte$ 
satisfying the claims in (\emph{a}2-2).
An \emph{overview} of the main steps in our proof for (\emph{a}2-2) is as 
follows (where the comments are enclosed within a pair of $(*$ and $*)$)\footnote{For faster implementation, 
in the loop of Step 2 we can do binary search for the least possible $\kappa$ 
over the range $\set{1,2,\dots,|\wte|}$ for which the polytope's optimal solution value is at most $\Gamma$, 
requiring $\lceil \log_2 (1 + |\wte|) \rceil$ iterations instead of $|\wte|$ iterations. For clarity, we omit 
such obvious improvements.}.
\begin{align*}
\text{\bf 1.} &
  \text{(* \emph{same as in \emph{(\emph{a}1)}} *)}
\\
                   &
  \text{We identify sub-cases whose optimal solutions are easy to find.} 
\\
                   &
  \text{Following steps apply only to all remaining sub-cases.} 
\\
\text{\bf 2.} &
  \text{\ffor\ $\kappa=1,2,\dots,|\wte|$ \ddo\ \hspace*{0.2in} (* assume 
	     $\mathsf{OPT}_{\mbox{\eadp}_{\mathfrak{C}^2_d}}(G,\wte,\gamma)=\kappa$ *)} 
\\
\text{\bf 2\emph{a}.} &
  \hspace*{0.2in}
  \text{(* \emph{as in \emph{(}a\emph{1}}) \emph{but with an additional constraint} *) $\,\,$}
\\
                   &
  \hspace*{0.2in}
  \text{we reduce our problem to a (directed) minimum $s$-$t$ cut problem} 
\\
                   &
  \hspace*{0.3in}
  \text{with the following additional constraint} 
\\
                   &
  \hspace*{0.4in}
  \text{$\blacktriangleright$ the number of edges to be deleted from $\wte$ is $\kappa$}
\\
                   &
  \hspace*{0.6in}
	\text{such that the following statements hold:}
\\
                   &
  \hspace*{0.6in}
  \text{$\triangleright$ The cut network can be constructed in polynomial time.}
\\
                   &
  \hspace*{0.6in}
  \text{$\triangleright$ There exists a feasible solution of \eadp$_{\mathfrak{C}^2_d}(G,\wte,\gamma)$}
\\
                   &
  \hspace*{0.75in} 
	\text{if and only if the minimum cut value is} 
\\
                   &
  \hspace*{0.75in} 
	\text{at most $\Gamma\eqdef\mathfrak{C}^2_d(G)-\gamma + |\wte|$.}
\\
\text{\bf 2\emph{b}.} &
  \hspace*{0.2in}
	\text{find an extreme-point optimal solution for an appropriate} 
\\
                          &
  \hspace*{0.35in}
	\text{polytope for the constrained minimum cut problem}
\\
                          &
  \hspace*{0.35in}
	\text{in polynomial time.}
\\
\text{\bf 2\emph{c}.} &
  \hspace*{0.2in}
	\text{\iif\ the optimal objective value is at most $\Gamma$ \tthen}
\\
\text{\bf 2\emph{c}(\emph{i}).} &
  \hspace*{0.4in}
	\text{carefully convert relevant fractional values in the solution} 
\\
                          &
  \hspace*{0.55in}
	\text{to integral values to get a solution in polynomial time.}
\\
\text{\bf 3.} &
   \text{Return the best among all solutions found in Step \textbf{2}}
\\
                          &
	 \hspace*{0.2in}
	 \text{as the desired solution.}
\end{align*}
%

\noindent
{\bf Step 2\emph{b}. Formulating an appropriate polytope for the constrained minimum cut problem}

\medskip
We showed in the proof of 
part (\emph{a}1)
that the feasibility problem can be reduced to finding a minimum $s$-$t$ cut of the directed graph $\cG=(\cV,\cE)$.
Notice that $\cG$ is acyclic, and every path between $s$ and $t$ has \emph{exactly three} directed edges, namely 
an edge-arc followed by a ed-cy-arc followed by a cycle-arc.
The minimum $s$-$t$ cut problem for a graph has a well-known associated convex polytope 
of \emph{polynomial} size (\EG, see~\cite[pp.\ 98-99]{V01}). 
Letting $p_\beta$ to be the variable corresponding to each node $\beta\in\cV$,
and $d_\alpha$ to be the variable associated with the edge $\alpha\in\cE$, 
this minimum $s$-$t$ cut polytope for the graph $\cG$ is as follows: 
\begin{gather}
\hspace*{-0.2in}
\text{
\begin{tabular}{r l}
{minimize} & $\sum_{\alpha\in\cE} \ccap(\alpha) d_\alpha$
\\
[3pt]
&
\hspace*{0.1in}$= \sum_{{\alpha\in\cE,\text{\footnotesize $\alpha$ is not ed-cy-arc} }} d_\alpha
+
\sum_{{\alpha\in\cE,\text{\footnotesize $\alpha$ is ed-cy-arc} }} \infty \times d_\alpha
$
\\
[4pt]
{subject to} & $d_\alpha \geq p_\beta-p_\xi$ \hspace*{0.15in} for every edge $\alpha=(\beta,\xi)\in\cE$ 
\\
[2pt]
                  & $p_s - p_t \geq 1$ 
\\
[2pt]
                  & $0\leq p_\beta \leq 1$ \hspace*{0.3in} for every node $\beta\in\cV$ 
\\
[2pt]
                  & $0\leq d_\alpha \leq 1$ \hspace*{0.3in} for every edge $\alpha\in\cE$ 
\end{tabular}
}
\label{pto1}
\end{gather}
It is well-known that all extreme-point solutions of~\eqref{pto1} are integral.
An integral solution of~\eqref{pto1} generates a $s$-$t$ cut $(\cS,\cV\setminus\cS)$ 
by letting $\cS=\set{\beta \,|\, p_\beta=1}$ and 
$\cut(\cS,\cV\setminus\cS)=\set{\alpha \,|\, d_\alpha=1}$.
For our case, we have an additional constraint in that 
the number of edges to be deleted from $\wte$ is $\kappa$,
which motivates us to formulate
the following polytope for our problem:
\begin{gather}
\hspace*{-0.2in}
\text{
\begin{tabular}{r l}
{minimize} & $\sum_{\alpha\in\cE} \ccap(\alpha) d_\alpha$
\\
[4pt]
&
\hspace*{0.1in}
$= \sum_{{\alpha\in\cE,\,\text{\footnotesize $\alpha$ is not ed-cy-arc} }} d_\alpha
+
\sum_{{\alpha\in\cE,\, \text{\footnotesize $\alpha$ is ed-cy-arc} }} \infty \times d_\alpha
$
\\
[4pt]
{subject to} & $d_\alpha \geq p_\beta-p_\xi$ \hspace*{0.15in} for every edge $\alpha=(\beta,\xi)\in\cE$ 
\\
[2pt]
                  & $p_s - p_t \geq 1$ 
\\
[2pt]
                  & $0\leq p_\beta \leq 1$ \hspace*{0.3in} for every node $\beta\in\cV$ 
\\
[2pt]
                  & $0\leq d_\alpha \leq 1$ \hspace*{0.3in} for every edge $\alpha\in\cE$ 
\\
[2pt]
             & $\sum_{u_e\in\cV} p_{u_e}=\kappa$ 
\end{tabular}
}
\label{pto2}
\end{gather}
Let $\mathsf{OPT}_{\eqref{pto2}}$ 
denote the optimal objective value of~\eqref{pto2}. 

\begin{lemma}\label{pto-obj}
$\mathsf{OPT}_{\eqref{pto2}} \leq \Gamma$. 
\end{lemma}

\begin{proof}
Suppose that removal 
of the edges in the optimal solution $\whe$ removes $\widehat{\Delta}\leq\Delta'$ cycles from $\cC'(G)$.
We construct the following solution of~\eqref{pto2} with respect to the optimal solution $\whe$ of 
\eadp$_{\mathfrak{C}^2_d}(G,\wte,\gamma)$ having $|\whe|=\kappa$ nodes: 
\begin{gather*}
\cS = \{s\} \bigcup \{ u_e \,|\, e\in \whe \} 
\bigcup 
\left\{
u_{\cF_i} \,|\, \text{$\cF_i$ contains at least one edge from $\whe$}
\right\}
\\
p_\beta = \left\{
\begin{array}{r l}
1, & \mbox{if $\beta\in\cS$}
\\
0, & \mbox{otherwise}
\end{array}
\right.
\,\,\,
\,\,\,
\,\,\,
d_\alpha = \left\{
\begin{array}{r l}
1, & \mbox{if $\alpha\in\cut(\cS,\cV\setminus\cS)$}
\\
0, & \mbox{otherwise}
\end{array}
\right.
\end{gather*}
It can be verified as follows that this is indeed a feasible solution
of~\eqref{pto2}:
\begin{itemize}
\item
Since $p_s=1$ and $p_t=0$, it follows that $p_s-p_t\geq 1$ is satisfied.
\item
No ed-cy-arc belongs to $\cut(\cS,\cV\setminus\cS)$. Thus, if $\alpha=(\beta,\xi)$ is an 
ed-cy-arc then $d_\alpha=0$ and it is not the case that $p_\beta=1$ and $p_\xi=0$.  
Thus for every ed-cy-arc $\alpha$ the constraint 
$d_\alpha \geq p_\beta-p_\xi$ is satisfied.
\item
Consider an edge-arc $\alpha=(s,u_e)$; note that $p_s=1$. If $u_e\in\cS$ that 
$p_{u_e}=1$ and $d_{\alpha}=0$, otherwise
$p_{u_e}=0$ and $d_{\alpha}=1$. In both cases, 
the constraint $d_\alpha \geq p_\beta-p_\xi$ is satisfied.
The case of a cycle-arc is similar.
\item
The constraint $\sum_{u_e\in\cV} p_{u_e}=\kappa$ 
is \emph{trivially} satisfied since $|\wte|=\kappa$ by our assumption.
\end{itemize}
Note that $\cut(\cS,\cV\setminus\cS)$ does \emph{not} contain any ed-cy-arcs.
Thus, the objective value of this solution is
\[
\sum_{a\in\cE} d_a
=
\sum_{a\in\cut(\cS,\cV\setminus\cS)} \hspace*{-0.2in} d_a
=
| \set{ u_e \,|\, u_e\notin\cS  } |
+
| \set{ u_{\cF_j} \,|\, u_{\cF_j} \in \cS  } |
=
(|\wte| - |\whe|) + \widehat{\Delta} \leq \Gamma
\]
%
where the last inequality follows by~\eqref{eq-mc} since 
$\mathfrak{C}^2_d(G\setminus \whe) \geq \gamma$.
\end{proof}

\medskip
\noindent
\text{\bf Step 2\emph{c}. Post-processing fractional values in the polytope solution}

\medskip
Given a polynomial-time obtainable optimal solution values
$\big\{d_\alpha^\ast,p_\beta^\ast \,|\, \alpha\in\cE, \, \beta\in\cV \big\}$ 
of the variables in~\eqref{pto2},
consider the following simple rounding procedure, 
the corresponding cut $(\cS,\cV\setminus\cS)$ of $\cG$,
and 
the corresponding solution $E_1\subseteq\wte$ of 
\eadp$_{\mathfrak{C}^2_d}(G,\wte,\gamma)$:
\begin{gather}
\hat{p}_\beta=\left\{
\begin{array}{r l}
1, & \mbox{if $p_\beta^\ast\geq\nicefrac{1}{2}$}
\\
0, & \mbox{otherwise}
\end{array}
\right.
\,\,\,\,\,\,\,\,\,
\cS=\set{\beta\in\cV \,|\, \hat{p}_\beta=1}
\,\,\,\,\,\,\,\,\,
E_1= \set{e \,|\, u_e\in\cS}
\label{rs1}
\end{gather}
Note that in inequalities 
$p_s - p_t \geq 1$, 
$0\leq p_s\leq 1$  
and 
$0\leq p_t\leq 1$  
ensures that $p_s^\ast=1$ and $p_t^\ast=0$.

\begin{lemma}
$|E_1|\leq 2\,\kappa$.
\end{lemma}

\begin{proof}
$
|E_1| = 
| \set{u_e \,|\, p_{u_e}^\ast\geq\nicefrac{1}{2} } |
\leq 
2\,\sum_{u_e\in\cV} p_{u_e}^\ast=2\,\kappa
$.
\end{proof}

\begin{lemma}\label{lemma-opt3}
$\ccap(\cut(\cS,\cV\setminus\cS)) \leq 2\, \mathsf{OPT}_{\eqref{pto2}} \leq 2\,\Gamma$.
\end{lemma}

\begin{proof}
Since $\ccap(\alpha)=\infty$ and 
$\mathsf{OPT}_{\eqref{pto2}}\leq\Gamma<\infty$, 
$d_\alpha^\ast=0$ for any ed-cy-arc $\alpha=(u_e,u_{\cF_j})$,
and thus $p_{u_e}^\ast \leq p_{u_{\cF_j}}^\ast$ for such an edge.
It therefore follows that 
\[
p_{u_e}^\ast \geq\nicefrac{1}{2}
\,\Rightarrow\,
p_{u_{\cF_j}}^\ast \geq\nicefrac{1}{2}
\,\,\,\,\,\equiv\,\,\,\,\,
\hat{p}_{u_e} =1 
\,\Rightarrow\,
\hat{p}_{u_{\cF_j}}=1
\]
Thus, no ed-cy-arc belongs to $\cut(\cS,\cV\setminus\cS)$.
Thus using Lemma~\ref{pto-obj} it follows that  
\begin{multline*}
\ccap(\cut(\cS,\cV\setminus\cS)) 
=
| \set{ (s,u_e) \,|\, \hat{p}_{u_e}=0  } |
+
| \set{ (u_{\cF_j},t) \,|\, \hat{p}_{u_{\cF_j}} =1   } |
\\
=
| \set{ (s,u_e) \,|\, p_{u_e}^\ast < \nicefrac{1}{2}  } |
+
| \set{ (u_{\cF_j},t) \,|\, p_{u_{\cF_j}}^\ast \geq \nicefrac{1}{2}   } |
\\
\leq
2\, \sum_{p_{u_e}^\ast < \nicefrac{1}{2}} (p_s^\ast - p_{u_e}^\ast) 
+
2\, \sum_{p_{\cF_j}^\ast \geq \nicefrac{1}{2}} (p_{\cF_j}^\ast-p_t^\ast) 
\\
\leq
2\, \sum_{p_{u_e}^\ast < \nicefrac{1}{2}} d_{s,p_{u_e}}^\ast
+
2\, \sum_{p_{\cF_j}^\ast \geq \nicefrac{1}{2}} d_{s,p_{\cF_j}}^\ast
<
2\,\sum_{\alpha\in\cE} \ccap(\alpha) d_\alpha^\ast
\leq 2\,\Gamma
\end{multline*}
\end{proof}

Since no ed-cy-arc belongs to $\cut(\cS,\cV\setminus\cS)$, if an edge $e\in E_1$ 
is involved in a cycle $\cF_j\in \cC'(G)$ then it must be the case 
that $(u_e,u_{\cF_j})\notin\cut(\cS,\cV\setminus\cS)$.
Thus, letting 
$m_1 = |E_1|$ and 
$\Delta_1 = |\,\set{ \cF_j\in \cC'(G) \,|\, u_{\cF_j}\in\cS}\,|$, 
the claimed bound on 
$\mathfrak{C}^2_d(G\setminus E_1)$ can be shown as follows using Lemma~\ref{lemma-opt3}:
\begin{gather*}
\begin{array}{r l}
& \ccap(\cut(\cS,\cV\setminus\cS))
=
(m-m_1)+\Delta_1
\leq 2\,\Gamma
= 2 \, \mathfrak{C}^2_d(G) -2\, \gamma + 2\,|\wte|
\\
[3pt]
\Rightarrow & 
\mathfrak{C}^2_d(G\setminus E_1) = 
\mathfrak{C}^2_d(G) +m_1 - \Delta_1
\geq
2\,\gamma -\mathfrak{C}^2_d(G) 
-(2\,|\wte|-m)
, \,\,\,\,\text{by~\eqref{eq-minus-e1}}
\\
[3pt]
\Rightarrow & 
\dfrac{\mathfrak{C}^2_d(G\setminus E_1) - \mathfrak{C}^2_d(G) }{ \gamma -\mathfrak{C}^2_d(G) }
\geq 
2 -  \dfrac{ 2\,|\wte|-m } { \gamma -\mathfrak{C}^2_d(G) }
\geq
2- \dfrac{1}{\frac{1}{2}+\eps}
=
\dfrac{4\eps}{1+2\eps}
\end{array}
\end{gather*}

\noindent
\textbf{(\emph{a}2-1)}
The decision version of computing 
$\mathsf{OPT}_{\mbox{\eadp}_{\mathfrak{C}^2_d}}(G,\wte,\gamma)$
is as follows: 
``given an instance \eadp$_{\mathfrak{C}^2_d}(G,\wte,\gamma)$
and an integer $\kappa>0$, is there a solution $\whe\subseteq\wte$ 
satisfying $|\whe|\leq\kappa$ ?''. 
We first consider the case of $d=3$.
We will reduce from the decision version of the
\dkst problem which is defined as follows.

\begin{definition}[\dkst problem]
Given an undirected graph $G_1=(V_1,E_1)$ where the degree of every node is either $2$ or $3$ 
and two integers $k$ and $t$, is there a (node-induced) subgraph of $G_1$ 
that has $k$ nodes and at least $t$ edges? 
\end{definition}

Assuming that their reduction is done from the clique problem on a $(n-4)$-regular $n$-node graph 
(which is $\NP$-hard~\cite{CC06}), 
the proof of Feige and Seltser in~\cite{FS97} shows that \dkst is $\NP$-complete for 
the following parameter values (for some integer $\sqrt{n}<\alpha\leq n-4$): 
\begin{gather*}
|V_1|=n^2+ (\alpha\,n+1) \left( \frac{n\,(n-4)}{2} \right),
\,\,\,\,
\,\,\,\,
|E_1|= |V_1|+ \frac{n \,(n-4)}{2}
\\
k = \alpha \,n + \binom{\alpha}{2}(\alpha\,n+1), 
\,\,\,\,
\,\,\,\,
t = \alpha \,n + \binom{\alpha}{2}(\alpha\,n+2)
\end{gather*}
We briefly review the reduction of Feige and Seltser in~\cite{FS97} 
as needed from our purpose. Their reduction is from the $\alpha$-CLIQUE problem which is defined 
as follows.

\begin{definition}[$\alpha$-CLIQUE problem]
Given a graph of $n$ nodes, does there exist a clique (complete subgraph) of size $\alpha$?
\end{definition}

Given an instance of $\alpha$-CLIQUE, they create an instance $G_1=(V_1,E_1)$ 
of \dkst (with the parameter values shown above) 
in which every node is replaced by a cycle of $n$ edges and an edge between two nodes is replaces 
by a path of length $\alpha\,n+3$ between two unique nodes of the two cycles corresponding to the two 
nodes (see \FI{fig3}(\emph{a})--(\emph{b}) for an illustration).
Given such an instance of \dkst with $V_1=\set{u_1,\dots,u_{|V_1|}}$ and 
$E_1=\set{a_1,\dots,a_{|E_1|}}$, we create an instance of 
\eadp$_{\mathfrak{C}^2_3}(G,\wte,\gamma)$ as follows: 
\begin{itemize}
\item
We associate each node $u_i\in V_1$ with a triangle  (the ``node triangle'') 
$\cL_i$ of $3$ nodes in $V$ such that every edge $\set{u_i,u_j}\in E_1$
is mapped to a \emph{unique} edge (the ``shared edge'') $e_{u_i,u_j}\in E$ that is shared by $\cL_i$ and $\cL_j$ 
(see \FI{fig3}(\emph{c})). 
Since in the reduction of 
Feige and Seltser~\cite{FS97}
all nodes have degree $2$ or $3$ and two degree $3$ nodes do not share more than one edge
such a node-triangle association is possible.
We set $\wte$ to be the set of \emph{all} shared edges; note that $|\wte|=|E_1|$.  
Let $\cL=\set{v_1,v_2,\dots\dots}$ be the set of all nodes in the that appear in any node triangle; note that 
$|\cL|< 3\,|V_1|$. 
\item
To maintain connectivity after all edges in $\wte$ are deleted, 
we introduce $3|\cL|+1$ new nodes 
$\set{w_0} \,\cup\, \big\{ w_{i,j} \,|\, i\in\{1,2,\dots,|\cL|\},\, j\in\{1,2,3\} \big\}$
and 
$4\,|\cL|$ new edges
\[
\Big\{ \set{w_0,w_{j,1}},\, \set{w_{j,1},w_{j,2}},\, \set{w_{j,2},w_{j,3}},\, \set{w_{j,3},v_j} \,|\, j\in\{1,2,\dots,|\cL|\} \Big\}
\]
\item
We set $\gamma=\mathfrak{C}^2_3(G)+(t-k)=\mathfrak{C}^2_3(G)+\binom{\alpha}{2}$.
\end{itemize}
First, we show that 
\eadp$_{\mathfrak{C}^2_3}(G,\wte,\gamma)$ indeed \emph{has} a trivial feasible solution, namely 
a solution that contains all the edges from $\wte$.
The number of triangles $\Delta'$ 
that include one or more edges from $\wte$ is 
precisely $|V_1|$ and thus using~\eqref{eq-minus-e1} we get: 
\begin{multline*}
\mathfrak{C}^2_3(G\setminus\wte)
=
\mathfrak{C}^2_3(G) + |\wte|-|\Delta'|
=
\mathfrak{C}^2_3(G) + |E_1|-|V_1|
\\
=
\mathfrak{C}^2_3(G) + \frac{n(n-4)}{2}
>
\mathfrak{C}^2_3(G)+\binom{\alpha}{2} =
\gamma
\end{multline*}
where the last inequality follows since $\alpha\leq n-4$.
The following lemma completes our proof.

\begin{lemma}
$G_1$ has a subgraph of $k$ nodes and at least $t$ edges
if and only if 
the instance of \eadp$_{\mathfrak{C}^2_3}(G,\wte,\gamma)$
constructed above
has a solution $\whe\subseteq\wte$ satisfying $|\whe|\leq t$.
\end{lemma}

\begin{proof}
Suppose that 
$G_1$ has $k$ nodes $u_1,u_2,\dots,u_k$ such that the subgraph $H_1$ induced by these nodes has $t'\geq t$ edges.
Remove an arbitrary set of $t'-t$ edges from $H_1$ to obtain a subgraph $H_1'=(V_1',E_1')$, 
and let $\whe=\{ e_{u_i,u_j} \,|\, i,j\in \{1,2,\dots,k\}, \, \{u_i,u_j\}\in E_1' \}$.
Obviously, $|\whe|=t$.
Consider the triangle $\cL_i$ corresponding to a node $u_i\in\{u_1,u_k,\dots,u_k\}$, and let 
$I(\cL_i)$ be the $0$-$1$ \emph{indicator variable} denoting if $\cL_i$ is eliminated by removing the edges in 
$\whe$, \IE, $I(\cL_i)=1$ (resp., $I(\cL_i)=0$) if and only if  
$\cL_i$ \emph{is} eliminated (resp., is \emph{not} eliminated) by removing the edges in $\whe$.
Note that the triangle $\cL_i$ gets removed if and only if there exists another node $u_j\in\{u_1,u_k,\dots,u_k\}$
such that $\{u_i,u_j\}\in E_1$. Thus, the total number of triangles
eliminated by removing the edges in $\whe$ is at most 
$\sum_{i=1}^k I(\cL_i)\leq k$ and consequently 
\[
\mathfrak{C}^2_3(G\setminus E')
=
\mathfrak{C}^2_3(G) + |\whe|-\sum_{i=1}^k I(\cL_i)
\geq
\mathfrak{C}^2_3(G) + t-k
=
\gamma
\]
Conversely, suppose that the instance of \eadp$_{\mathfrak{C}^2_3}(G,\wte,\gamma)$
has a solution $\whe\subseteq\wte$ satisfying $|\whe|=\widehat{t}\leq t$.
Let 
\\
$V_1'=\{ u_j \,|\, \cL_j \text{ is removed by removing one of more edges from } \whe \}$.
Using~\eqref{eq-minus-e1} we get
\begin{gather}
\mathfrak{C}^2_3(G\setminus \whe) \geq \gamma = \mathfrak{C}^2_3(G) + t-k
\,\Rightarrow\,
\widehat{t} - |V_1'| \geq t-k
\label{ggg1}
\end{gather}
Let $H_1'=(V_1',E_1')$ be the subgraph of $G_1$ induced by the nodes in $V_1'$.  
Clearly, $|E_1'|\geq \widehat{t}$.
If $|V_1|<k$ then 
we use the following procedure to add $k-|V_1'|$ nodes: 

\begin{center}
\begin{tabular}{l}
$V_1''\leftarrow V_1'$ 
\\
\wwhile\ $|V_1''|\neq k$ \ddo
\\
\hspace*{0.2in}
select a node $u_j\notin V_1''$ connected to one or more nodes in $V_1''$, 
\\
\hspace*{0.3in}
and add $u_j$ to $V_1''$
\end{tabular}
\end{center}

\noindent
Let $H_1''=(V_1'',E_1'')$ be the subgraph of $G_1$ induced by the nodes in $V_1''$.  
Note that $|V_1''|=k$ and $|E_1''|\geq |E_1'|+(k-|V_1'|)$, and thus using~\eqref{ggg1} we get
\[
|E_1''|
\geq 
|E_1'| + (k-|V_1'|)
\geq 
\widehat{t} + (k-|V_1'|)
\geq t
\]
\end{proof}

This concludes the proof for $d=3$. For the case when $d=3\mu$ for some integer $\mu>1$,
the same reduction can be used provide we split \emph{every} edge of $G_1$ into a path of length $\mu$ 
by using new $\mu-1$ nodes 
(see \FI{fig3}(\emph{d})).

\bigskip
\noindent
\textbf{(\emph{b}1) and (\emph{b}2)}
We \emph{will} reuse the notations used in the proof of (\emph{a}).
We modify the proof and the proof technique in 
\textbf{(\emph{a}1)} 
for the proof of 
\textbf{(\emph{b}1)}.
We now observe that the following sub-cases are easy to solve:
\begin{itemize}
\item
If $\gamma<n-m+1+\Delta-\Delta'$ 
then we can assert that there is no feasible solution. 
This is true because 
for any 
$E'\subseteq \wte$ it is true that
$\mathfrak{C}^2_d(G\setminus E')$ is 
at least $n-(m-1)+(\Delta-\Delta')$.
\item
If 
$\gamma\geq n-m+1+\Delta-\Delta'$ 
and 
$\Delta'=0$ then 
there exists a \emph{trivial} optimal feasible solution of the following form:  
{\em 
select any set of $m_1$ edges from $\wte$ 
where 
$m_1$ is the largest positive integer satisfying 
$n-m_1+1+\Delta\leq\gamma$.  
}
\end{itemize}
%
Thus, we assume that 
$\gamma\geq n-m+1+\Delta-\Delta'$ 
and $\Delta'>0$.
\eqref{eq-minus-e1} still holds, but 
\eqref{eq-mc} is now rewritten as (note that $\Gamma>0$):
\begin{multline}
\mathfrak{C}^2_d(G\setminus E_1) \leq \gamma
\,\equiv\,
m_1-\Delta_1 \leq  \gamma-\mathfrak{C}^2_d(G)
\\
\equiv\,
m_1+ (\Delta'-\Delta_1) \leq  \gamma-\mathfrak{C}^2_d(G) +\Delta'
	 \eqdef \Gamma
\tag*{\eqref{eq-mc}$'$}
\end{multline}
The nodes in the di-graph $\cG=(\cV,\cE)$ are same as before, but 
the directed edges are modified as follows:
\begin{itemize}
\item
For every edge $e\in \wte$, 
we have an edge $(u_e,t)$ (an ``edge-arc'') of capacity 
$\ccap(u_e,t)=1$.
\item
For every cycle $\cF_i\in\cC'(G)$, 
we have an edge (a ``cycle-arc'') $(s,u_{\cF_i})$ of capacity 
$\ccap(s,u_{\cF_i})=1$.
\item
For every cycle $\cF_i\in\cC'(G)$ and 
every edge $e\in \wte$ such that $e$ is an edge of $\cF_i$, we have a directed edge
(an ``cycle-edge-arc'', cy-ed-arc for short) $(u_{\cF_i},u_e)$ of capacity $\ccap(u_{\cF_i},u_e)=\infty$.
\end{itemize}
Corresponding to a feasible solution 
$E_1$ of $m_1$ edges 
for 
\eadp$_{\mathfrak{C}^2_d}(G,\wte,\gamma)$
that removes
$\Delta_1$ cycles, 
exactly the same cut 
$(\cS,\cV\setminus\cS)$ described before 
includes \emph{no} cy-ed-arcs
and has a capacity of 
\begin{eqnarray*}
\ccap(\cut(\cS,\cV\setminus\cS))
& = & 
| \set{(s,u_{\cF_i}) \,|\, \text{$\cF_i$ does not contain one or mores edges from $E_1$} } |
\\
& & 
\,\,\,+\,
| \set{(u_e,t) \,|\, e\in E_1} |
\\
& = &
(\Delta'-\Delta_1)+ m_1 
\end{eqnarray*}
Therefore 
$\mathfrak{C}^2_d(G\setminus E_1) \leq \gamma$ implies
$\ccap(\cut(\cS,\cV\setminus\cS)) \leq \Gamma$,
as desired. 
Conversely,
given a minimum $s$-$t$ cut 
$(\cS,\cV\setminus\cS)$ of $\cG$ of value $\Phi\leq\Gamma$, 
we consider the solution 
$E_1 = \{ e \,|\, u_e\in\cS \}$
for 
\eadp$_{\mathfrak{C}^2_d}(G,\wte,\gamma)$.
Let 
$\Psi = \{ \cF_j \,|\, u_{\cF_j}\in\cS \}$
and let 
$\Upsilon$ be the 
cycles from $\cC'(G)$ that are removed by deletion of the edges in $E_1$.
Since \emph{no} cy-ed-arc (of infinite capacity) can be an edge of the minimum $s$-$t$ cut 
$(\cS,\cV\setminus\cS)$, 
$\Psi$ is a subset of $\Upsilon$. 
We therefore have 
\begin{multline*}
\Phi 
=
| \set{ u_{\cF_j} \,|\, u_{\cF_j} \notin \cS  } |
+
| \set{ u_e \,|\, u_e\in\cS  } |
\\
= 
(\Delta'-|\Psi|) + |E_1|
\leq 
\Gamma 
\,\Rightarrow\,
(\Delta'-|\Upsilon|) + |E_1|
\leq 
\Gamma 
\end{multline*}
and the last inequality implies $\mathfrak{C}^2_d(G\setminus E_1) \leq \gamma$.

This completes a proof for (\emph{b}1). We now prove (\emph{b}2).
We use an approach similar to that in 
(\emph{a}2) but with a \emph{different polytope} for the minimum $s$-$t$ cut of $\cG$.
Let $\cP$ be the set of all possible $s$-$t$ paths in $\cG$. Then, an alternate polytope
for the minimum $s$-$t$ cut is as follows (cf.\ see (20.2) in~\cite[p. 168]{V01}): 
\begin{gather}
\text{
\begin{tabular}{r l}
{minimize} & $\sum_{\alpha\in\cE} \ccap(\alpha) d_\alpha$
\\
[4pt]
{subject to} & $\sum_{\alpha\in p} d_\alpha \geq 1$ \hspace*{0.15in} for every $s$-$t$ path $p\in\cP$ 
\\
[2pt]
                  & $0\leq d_\alpha \leq 1$ \hspace*{0.3in} for every edge $\alpha\in\cE$ 
\end{tabular}
}
\tag*{\eqref{pto1}$'$}
\end{gather}
An integral solution of~\eqref{pto1}$'$ generates a $s$-$t$ cut $(\cS,\cV\setminus\cS)$ 
by letting $\cut(\cS,\cV\setminus\cS)=\set{\alpha \,|\, d_\alpha=1}$.
Since the capacity of any cy-ed-arc in $\infty$, 
$\cut(\cS,\cV\setminus\cS)$
contains only cycle-arcs or edge-arcs, and the number of edge-arcs in 
$\cut(\cS,\cV\setminus\cS)$
for an integral solution is precisely 
the number of edge-nodes in $\cS$.  
This motivates us to formulate
the following polytope for our problem to ensure that 
integral solutions constrain 
the number of edges to be deleted from $\wte$ to be $\kappa$: 
\begin{gather}
\text{
\begin{tabular}{r l}
{minimize} & $\sum_{\alpha\in\cE} \ccap(\alpha) d_\alpha$
\\
[4pt]
{subject to} & $\sum_{\alpha\in p} d_\alpha \geq 1$ \hspace*{0.15in} for every $s$-$t$ path $p\in\cP$ 
\\
[2pt]
                  & $0\leq d_\alpha \leq 1$ \hspace*{0.3in} for every edge $\alpha\in\cE$ 
\\
[2pt]
                  & $\sum_{e\in\wte}d_{(u_e,t)}=\kappa$ 
\end{tabular}
}
\tag*{\eqref{pto2}$'$}
\end{gather}
For our problem, $|\cP|<\binom{|\cV|}{3}$ and 
thus~\eqref{pto2}$'$ can be solved in polynomial time. 
Let $\mathsf{OPT}_{\eqref{pto2}'}$ 
denote the optimal objective value of~\eqref{pto2}$'$.
It is very easy to see that 
$\mathsf{OPT}_{\eqref{pto2}'} \leq \Gamma$: 
assuming that deletion 
of the $\kappa$ edges in the optimal solution $\whe$ removes $\widehat{\Delta}$ cycles from $\cC'(G)$, 
we set
$
d_\alpha = \left\{
\begin{array}{r l}
1, & \mbox{$\alpha=(u_{\cF_j},u_e),\,e\in\whe$}
\\
0, & \mbox{otherwise}
\end{array}
\right.
$
to construct
a feasible solution of~\eqref{pto2}$'$
of objective value 
\[
\sum_{\alpha\in\cE} d_\alpha
=
| \set{ u_{\cF_j} \,|\, d_{(s,u_{\cF_j})} = 1  } |
+
| \set{ u_e \,|\, d_{(u_e,t)}=1  } |
=
(\Delta' - \widehat{\Delta}) + |\wte| \leq \Gamma
\]
where the last inequality follows by~\eqref{eq-mc}$'$ since 
$\mathfrak{C}^2_d(G\setminus \whe) \leq \gamma$.
Note that the constraint 
$\sum_{e\in\wte}d_{(u_e,t)}=\kappa$ 
is satisfied since 
$\sum_{e\in\wte}d_{(u_e,t)}
=
| \set{d_{(u_e,t)} \,|\, d_{(u_e,t)}=1} |
=
| \set{e \,|\, e\in\whe } |
=\kappa$. 

Given a polynomial-time obtainable optimal solution values
$\big\{d_\alpha^\ast \,|\, \alpha\in\cE \big\}$ 
of the variables in~\eqref{pto2}$'$,
consider the following simple rounding procedure, 
the corresponding cut $(\cS,\cV\setminus\cS)$ of $\cG$,
and 
the corresponding solution $E_1\subseteq\wte$ of 
\eadp$_{\mathfrak{C}^2_d}(G,\wte,\gamma)$:
\begin{gather}
\hat{d}_\alpha=\left\{
\begin{array}{r l}
1, & \mbox{if $d_\alpha^\ast\geq\nicefrac{1}{2}$}
\\
0, & \mbox{otherwise}
\end{array}
\right.
\,\,\,\,\,\,\,\,\,
E'= \set{\alpha \,|\, \hat{d}_\alpha=1} 
\,\,\,\,\,\,\,\,\,
E_1= \set{e \,|\, (u_e,t)\in E'}
\tag*{\eqref{rs1}$'$}
\end{gather}

\begin{lemma}
$E'$ is indeed a $s$-$t$ cut of $\cG$ and $E'$ does not contain any cy-ed-arc.
\end{lemma}

\begin{proof}
Since the capacity of any cy-ed-arc $\alpha$ in $\infty$, 
$d_\alpha^\ast=0$ and therefore $\alpha\notin E'$.
To see that $E'$ is indeed a $s$-$t$ cut,
consider any $s$-$t$ path $(s,u_{\cF_j})$,$(u_{\cF_j},u_e)$, $(u_e,t)$.
Since $d_{(u_{\cF_j},u_e)}^\ast=$, we have 
$
d_{(s,u_{\cF_j})}+d_{(u_{\cF_j},u_e)}+d_{(u_e,t)}=
d_{(s,u_{\cF_j})}+d_{(u_e,t)}\geq 1
$, which implies $\max\{d_{(s,u_{\cF_j})},d_{(u_e,t)}\}\geq\nicefrac{1}{2}$, putting
at least one edge of the path in $E'$ for deletion.
\end{proof}

Note that
$
|E_1|=
| \set{e \,|\, d_{(u_e,t)}^\ast\geq\nicefrac{1}{2} } |
\leq
2\,\sum_{e\in\wte}d_{(u_e,t)}^\ast=2\,\kappa
$, as desired. 
Let $(\cS,\cV\setminus\cS)$ be the $s$-$t$ cut such that $\cut(\cS,\cV\setminus\cS)=E'$.
It thus follows that 
\begin{multline}
\ccap(\cut(\cS,\cV\setminus\cS)) 
=
|E'|
=
|\set{\alpha \,|\, d_\alpha^\ast \geq \nicefrac{1}{2} }|
\\
\leq 
2\, \sum_{\alpha\in\cE} \ccap(\alpha) d_\alpha^\ast
=
2\, \mathsf{OPT}_{\eqref{pto2}'}
\leq
2\,\Gamma
\label{eq-less-gamma}
\end{multline}
Let 
$\Psi = \{ \cF_j \,|\, u_{\cF_j}\in\cS \}$
and let 
$\Upsilon$ be the 
cycles from $\cC'(G)$ that are removed by deletion of the edges in $E_1$.
Since \emph{no} cy-ed-arc (of infinite capacity) can be an edge of the minimum $s$-$t$ cut 
$(\cS,\cV\setminus\cS)$, 
$\Psi$ is a subset of $\Upsilon$.
The claimed bound on 
$\mathfrak{C}^2_d(G\setminus E_1)$ can now be shown as follows using \eqref{eq-less-gamma}:
\begin{gather*}
\begin{array}{r l}
& \ccap(\cut(\cS,\cV\setminus\cS))
\\
&
\, =
| \set{ u_{\cF_j} \,|\, u_{\cF_j} \notin \cS  } |
+
| \set{ u_e \,|\, u_e\in\cS  } |
\\
&
\, = 
(\Delta'-|\Psi|) + |E_1| \leq 2\,\Gamma 
\\
[3pt]
\Rightarrow & 
(\Delta'-|\Upsilon|) + |E_1|
\leq 
2\,\Gamma
=
2\,\gamma-2\,\mathfrak{C}^2_d(G) +2\,\Delta'
\\
[3pt]
\Rightarrow & 
\mathfrak{C}^2_d(G\setminus E_1) - \mathfrak{C}^2_d(G) 
= |E_1| - |\Upsilon|
\leq
2\,\gamma-2\,\mathfrak{C}^2_d(G) +\Delta'
\\
[3pt]
\Rightarrow & 
\frac { \mathfrak{C}^2_d(G\setminus E_1) - \mathfrak{C}^2_d(G) }
      { \gamma-\mathfrak{C}^2_d(G) }
\leq 
2 + \frac{ \Delta' } { \gamma-\mathfrak{C}^2_d(G) }
\leq 
1-\eps
\end{array}
\end{gather*}

\noindent
\textbf{(\emph{b}3)}
In the proof of Theorem~\ref{thm-badp}, 
set 
$\wte=E_2\setminus E_1$ and 
$\gamma=\mathfrak{C}^2_d(G_2)$. 
Note that the proof shows that 
$\mathfrak{C}^2_d(G_2)<\mathfrak{C}^2_d(G_1)<$ 
The proof also shows that 
$\mathfrak{C}^2_d(G_1\setminus E_3)\geq \gamma$ 
for any proper subset of edges $E_3\subset\wte$, which ensures that 
for any subset of edges $E_4\subseteq\wte$ 
$\mathfrak{C}^2_d(G_1\setminus E_4)\leq \gamma$ is equivalent to stating 
$\mathfrak{C}^2_d(G_1\setminus E_4)=\gamma$.


\subsection{Gromov-hyperbolic curvature: computational complexity of \eadp$_{\curvgrom}$} 
\label{sec-proof-extreme-gromov}

\begin{theorem}\label{thm-eadp-curvgrom}
The following statements hold for \eadp$_{\curvgrom}(G,\wte,\gamma)$ 
when $\gamma>\curvgrom(G)$:
\medskip
\begin{adjustwidth}{0.1in}{}
\begin{description}
\item[\emph{(}a\emph{)}]
Deciding if there exists a feasible solution is $\NP$-hard.
\item[\emph{(}b\emph{)}]
Even if a trivial feasible solution exists,
it is $\NP$-hard to design a polynomial-time 
algorithm to approximate 
$\mathsf{OPT}_{\mbox{\eadp}_{\curvgrom}}(G,\wte,\gamma)$
within a factor of $c\,n$ 
for some constants $c>0$, 
where $n$ is the number of nodes in $G$.
\end{description}
\end{adjustwidth}
%
\end{theorem}


\subsubsection{Proof of Theorem~\ref{thm-eadp-curvgrom}}
\label{sec-app-proof-thm-eadp-curvgrom}

From a high level point of view, Theorem~\ref{thm-eadp-curvgrom} is proved by suitably modifying the reductions used
in the proof of Theorem~\ref{thm-grom-hard}.

To prove
\textbf{(\emph{a})}
we will use a simpler version of the proof of Theorem~\ref{thm-grom-hard}
reusing the same notations. 
Our graph $G$ will be the same as the graph $G_1$ in that proof, except that 
we do \emph{not} add the  
complete graph $K_{|V''|}$ on the nodes $w_0,w_1,\dots,w_{|V''|-1}$
and consequently we also do \emph{not} have the edge $\set{u,w_0}$.
We set $\wte=E'$ and $\gamma=\frac{n}{2}+1$. 
The proof of Theorem~\ref{thm-grom-hard} shows that 
$\curvgrom(G)<\gamma$,
$\curvgrom(G\setminus E')\leq\gamma$ 
for any subset of edges $E'\subseteq\wte$, 
and 
$\curvgrom(G\setminus E')=\gamma$ 
for a subset of edges $\emptyset\subset E'\subset \wte$ 
if and only if the given cubic graph has a Hamiltonian path between the two specified nodes, 
thereby showing $\NP$-hardness of the feasibility problem.

To prove \textbf{(\emph{b})} the same construction in the proof of Theorem~\ref{thm-grom-hard}
works: $G$ is the same as the graph $G_1$ in that proof, 
$\gamma=\frac{n}{2}+1$, $\wte$ is the set of edges whose deletion produced $G_2$ from $G_1$,
and the trivial feasible solution is $G_2$.
Note that the proof of Theorem~\ref{thm-grom-hard} shows 
$\curvgrom(G)<\gamma$,
$\curvgrom(G\setminus E')\leq\gamma$ 
for any subset of edges 
$\emptyset\subset E'\subseteq\wte$ and 
$\curvgrom(G_2)=\gamma$.


\section{Computational complexity of targeted anomaly detection problems}
\label{sec-proof-target}

\subsection{Geometric curvatures: computational hardness of \badp$_{\mathfrak{C}^2_d}(G_1,G_2)$} 
\label{sec-proof-target-geom}

For two functions $f(n)$ and $g(n)$ of $n$, we say 
$f(n)=O^\ast(g(n))$ if $f(n)=O(g(n)\,n^c)$ for some positive constant $c$.
In the sequel we will use the following two complexity-theoretic assumptions:
the \emph{unique games conjecture} (\UGC)~\cite{K02,T12}, 
and the \emph{exponential time hypothesis} (\ETH)~\cite{IP01,IPZ01,W03}.

\begin{theorem}\label{thm-badp}~\\
\textbf{\emph{(\emph{a})}}
Computing $\mathsf{OPT}_{\mbox{\badp}_{\mathfrak{C}^2_3}}(G_1,G_2)$ is $\NP$-hard. 

\noindent
\textbf{\emph{(\emph{b})}}
There are no algorithms of the following type for 
\badp$_{\mathfrak{C}^2_d}(G_1,G_2)$ 
for $4\leq d\leq o(n)$
when $G_1$ and $G_2$ are $n$-node graphs: 
\medskip
\begin{adjustwidth}{0.1in}{}
\begin{description}
\item[\emph{(\emph{b}1)}]
a polynomial time $(2-\eps)$-approximation algorithm
for any constant $\eps>0$ assuming \UGC\ is true, 
\item[\emph{(\emph{b}2)}]
a polynomial time $(10\sqrt{5}-21-\eps)\approx 1.36$-approximation algorithm 
for any constant $\eps>0$
assuming P$\neq\NP$, 
\item[\emph{(\emph{b}3)}]
a $O^\ast\big(2^{o(n)}\big)$-time exact computation algorithm
assuming \ETH\ is true, and 
\item[\emph{(\emph{b}4)}]
a $O^\ast\big(n^{o(\kappa)}\big)$-time exact computation algorithm
if $\mathsf{OPT}_{\mbox{\badp}_{\mathfrak{C^2_3}}}(G_1,G_2)\leq \kappa$
assuming \ETH\ is true.
\end{description}
\end{adjustwidth}
\end{theorem}

\subsubsection{Proof techniques and relevant comments regarding Theorem~\ref{thm-badp}}
\label{sec-informal-badp}

\paragraph{\em\bf(on proof of \textbf{{(\emph{a})}})}
We prove the results by reducing the triangle deletion problem (\tdp) to that of solving 
\badp$_{\mathfrak{C}^2_3}$.
\tdp\ was shown to be $\NP$-hard by Yannakakis in~\cite{Yannakakis}. 

\paragraph{\em\bf(on proof of \textbf{{(\emph{b})}})} 
We provide suitable approximation-preserving reductions from \mnc.

\paragraph{\em\bf{(on proofs of \textbf{(\emph{b}3)} and \textbf{(\emph{b}4)})}}
For these proofs, the idea is to start 
with an instance of \tSAT, use 
``sparsification lemma'' in~\cite{IPZ01} 
to generate a family of Boolean formulae, reduce each of these formula 
to \mnc, and finally reduce each such instance of \mnc\
to a corresponding instance of 
\badp$_{\mathfrak{C}^2_d}$.

\subsubsection{Proof of Theorem~\ref{thm-badp}}
\label{sec-app-proof-thm-badp}

The \emph{minimum node cover} problem (\mnc) is defined as follows.

\begin{definition}[minimum node cover problem (\mnc)]
Given a graph $G$,
select a subset of nodes of \emph{minimum} cardinality such that at least one end-point of \emph{every} edge 
has been selected.
\end{definition}

Let $\optmnc(G)$ denote the cardinality of the subset of nodes that is an optimal solution of \mnc. 
The (standard) Boolean satisfiability problem is denoted by \SAT, and 
its restricted case when every clause has exactly $k$ literals
will be denoted by \kSAT~\cite{GJ79}.
Consider \SAT or \kSAT and let $\Phi$ be an input instance (\IE, a 
Boolean formula in conjunctive normal form) of it. 
The following inapproximability results are known for \mnc:
\begin{description}
\item[($\star_{\text{\mnc}}$)]
There exists a polynomial time algorithm that transforms a given instance $\Phi$ of \SAT 
to an input instance graph $G=(V,E)$ of \mnc\ such that
the following holds
for any constant $0<\eps<\frac{1}{4}$, assuming \UGC to be true~\cite{KR08}:
\[
\begin{array}{r l}
\text{if $\Phi$ is satisfiable then} & \optmnc(G)\leq \left( \frac{1}{2}+\eps\right) |V|
\\
\\
[-10pt]
\text{if $\Phi$ is not satisfiable then} & \optmnc(G)\geq \left( 1 - \eps\right) |V|
\end{array}
\]
\item[($\star\star_{\text{\mnc}}$)]
There exists a polynomial time algorithm that transforms a given instance $\Phi$ of \SAT 
to an input instance graph $G=(V,E)$ of \mnc\ such that
the following holds for any constant $0<\eps<16-8\sqrt{5}$ and 
for some $0<\alpha<2|V|$, assuming P$\neq\!\!\NP$~\cite{DS05}:
\[
\begin{array}{r l}
\text{if $\Phi$ is satisfiable then} & \optmnc(G)\leq \left(\frac{\sqrt{5}-1}{2}+\eps\right)\alpha
\\
\\
[-10pt]
\text{if $\Phi$ is not satisfiable then} & \optmnc(G)\geq \left( \frac{71-31\sqrt{5}}{2}-\eps\right)\alpha
\end{array}
\]
(note that $\left( { \frac{71-31\sqrt{5}}{2} } \right) \,/\, \left( { \frac{\sqrt{5}-1}{2} } \right) = 10\sqrt{5}-21\approx 1.36$).
\item[($\star\!\star\!\star_{\text{\mnc}}$)]
There exists a polynomial time algorithm (\EG, see~\cite[page 54]{GJ79})
that transforms a given instance $\Phi$ of \tSAT of $n$ variable and $m$ clauses to 
to an input instance graph $G=(V,E)$ of \mnc\ with $|V|=3n+2m$ nodes and $|E|=n+m$ edges such that
such that $\Phi$ is satisfiable if and only if $\optmnc(G)=n+2m$.
\end{description}
%

\medskip
\noindent
\textbf{Proof of (\emph{a})}
We will prove the results by reducing the triangle deletion problem to that of computing 
\badp$_{\mathfrak{C}^2_3}$.
The \emph{triangle deletion problem} (\tdp) can be stated as follows: 
\emph{Given $G$ find the minimum number of edges (which we will denote by $\opttdp(G)$) 
to be deleted from $G$ to make it triangle-free}. 
\tdp\ was shown to be $\NP$-hard by Yannakakis in~\cite{Yannakakis}. 

Consider an instance $G=(V,E)$ of \tdp\ 
where $V=\set{u_1,\dots,u_n}$ and $E=\set{e_1,\dots,e_m}$.
We create an instance 
$G_1=(V',E_1)$ and $G_2=(V',E_2)$ (with $\emptyset\subset E_2\subset E_1$)
of 
\badp$_{\mathfrak{C}^2_3}$
in the following manner:
\begin{enumerate}[label=$\triangleright$,leftmargin=0.7cm]
\item 
For each $u_i\in V$, we create a node $v_i\in V'$.
There are $n$ such nodes in $V'$.
\item 
If $\set{u_i,u_j}\in E$, then we add the edge $\set{v_i,v_j}$ to $E_1$. We
call these edges as 
``original''
edges. Let $E_d$ be the set of all original edges; note that $|E_d|=m$.
\item 
To ensure that $G_2$ is a connected graph, we add two new nodes $w_i^1,w_i^2$ in $V'$ 
corresponding to each node $v_i\in V'$ for $i=1,2,\dots,n-1$, 
and add three new edges 
$\set{v_i, w_i^1}$, $\set{w_i^1, w_i^2}$ and $\set{w_i^2, v_{i+1}}$ in $E_1$. 
This step adds $2n-2$ new nodes and $3n-3$ new edges to $V_1$ and $E_1$, respectively.
We call the new edges added in this step as 
``connectivity''
edges.
\item 
For each $\set{u_i,u_j})\in E$, we create a new node $v_{i,j}$ in $V'$ 
and add two new edges $\set{u_i,v_{ij}}$ and $\set{v_{ij}, u_j}$ to $E_1$. 
This step creates a new triangle corresponding to each original edge. We call the new edges added in this step as 
``triangle-creation''
edges.
This step adds $m$ new nodes and $2m$ new edges to $V_1$ and $E_1$, respectively, and exactly $m$ new triangles.
\end{enumerate}
Define $E_2=E_1\setminus E_d$. Thus, we have 
$|V'|=3n+m-2$, 
$|E_1|=3n+3m-3$, $|E_2|=3n+2m-3$, and $G_2$ contains \emph{no} triangles.
Let $\Delta$ is the number of triangles in $G_1$ created using only original edges (the ``original triangles''); 
note that $\Delta$ is also equal to the number of triangles in $G$. 
Then, $\mathfrak{C^2_3}(G_1) = |V'|-(3m+3n-3)+(\Delta+m)$ and $\mathfrak{C^2_3}(G_2) = |V'|-(2m+3n-3)$.
The following lemma completes our $\NP$-hardness proof.

\begin{lemma}\label{jj1}
$\opttdp(G)= \mathsf{OPT}_{\mbox{\badp}_{\mathfrak{C}^2_3}}(G_1,G_2)$.
\end{lemma}

\begin{proof}

\medskip
\noindent
\emph{Proof of} $\opttdp(G)\geq\mathsf{OPT}_{\mbox{\badp}_{\mathfrak{C}^2_3}}(G_1,G_2)$.

\medskip
Let $\eopt\subset E$ be an optimum solution of \tdp\ on $G$, 
\\
let 
$\eopt'=\set{\{v_i,v_j\}\,|\,\{u_i,u_j\}\in E}\subseteq E_d$, 
and consider the graph 
$G_3=(V',E_1\setminus\eopt')$.
Note that 
$G_3$ has \emph{no} original triangles and has exactly 
$m-|\eopt'|$ triangles involving triangle-creation edges, and thus 
\[
\mathfrak{C^2_3}(G_3)=
|V'| - (3n+3m-3-|\eopt'|) + (m-|\eopt'|)
=
\mathfrak{C^2_3}(G_2)
\]
and therefore 
$\mathsf{OPT}_{\mbox{\badp}_{\mathfrak{C}^2_3}}(G_1,G_2) \leq |\eopt'|= |\eopt|=\opttdp(G)$.

\medskip
\noindent
\emph{Proof of} $\opttdp(G)\leq\mathsf{OPT}_{\mbox{\badp}_{\mathfrak{C}^2_3}}(G_1,G_2)$.

\medskip
Suppose that  
$\eopt'\subset E_d$ is an optimum solution of $q$ edges of \badp$_{\mathfrak{C}^2_3}$ on $G_1$ and $G_2$,
let $G_3=(V',E_1\setminus\eopt')$ be the graph obtained from $G_1$ by removing the edges in $\eopt'$,  
and let $E'=\set{\{u_i,u_j\}\,|\,\{v_i,v_j\}\in \eopt'}\subseteq E$. 
Let $q=|\eopt'|$, $e_1',e_2',\dots,e_q'$ be an arbitrary ordering of the edges in $\eopt'$ 
and
$\delta_i'$ (for $i=1,2,\dots,q$)
is the number of triangles in $G_1$ that contains the edge $e_i'$ but \emph{none}
of the edges $e_1',\dots,e_{i-1}'$. 
Note that, for each $i$, exactly $\delta_i'-1$ triangles out of the $\delta_i'$ triangles 
are original triangles.
Let $\Delta'\leq\Delta$ be the number of original triangles removed by removing the edges in 
$\eopt'$; thus, 
$\Delta'=\sum_{i=1}^q \left(\delta_i'-1\right)$.
Simple calculations now show that 
\begin{multline*}
\mathfrak{C^2_3}(G_3) = 
|V'| - (3n+3m-3-|\eopt'|) + \left(\Delta+m-\sum_{i=1}^q \delta_i\right)
\\
=
|V'| - (3n+3m-3-|\eopt'|) + \left(\Delta+m-q-\sum_{i=1}^q \left( \delta_i-1 \right)\right)
\\
=
|V'| - (3n+3m-3-|\eopt'|) + \left(\Delta+m-|\eopt'|-\Delta'\right)
\\
=
|V'| - (3n+2m-3) + (\Delta-\Delta')
\end{multline*}
Consequently, $\mathfrak{C^2_3}(G_3) = \mathfrak{C^2_3}(G_2)$ implies $\Delta'=\Delta$ 
and $E'$ is a valid solution of \tdp\ on $G$.
This implies 
$\opttdp(G)\leq |E'|=|\eopt'|= \mathsf{OPT}_{\mbox{\badp}_{\mathfrak{C}^2_3}}(G_1,G_2)$.
\end{proof}

\bigskip
\noindent
\textbf{Proofs of (\emph{b}1) and (\emph{b}2)}

\medskip
Consider an instance graph $G=(V,E)$ of \mnc\ 
with $n$ nodes and $m$ edges 
where $V=\set{v_1,v_2,\dots,v_n}$ and $E=\set{e_1,e_2,\ldots,e_m}$. 
Let $\emptyset\subset\vmnc\subset V$ be an optimal solution 
of $\optmnc(G)=|\vmnc|$ nodes 
for this instance of \mnc.
We then create an instance
$G_1=(V',E_1)$ and $G_2=(V',E_2)$ (with $\emptyset\subset E_2\subset E_1$)
of 
\badp$_{\mathfrak{C}^2_d}$ for a given 
$d\geq 4$
in the following manner: 
\begin{itemize}
\item 
For each $v_i\in V$, we create $d$ new nodes $\set{v_i^1,v_i^2,\dots,v_i^d}$ in $V'$, 
and a $d$-cycle containing the edges
$\set{v_i^1,v_i^2}$,$\set{v_i^2,v_i^3}$,$\dots$,$\set{v_i^{d-1},v_i^d},\set{v_i^d,v_i^1}$ in $E_1$. 
We call the cycles generated in this step as the ``node cycles''.
This creates a total of $dn$ nodes in $V'$ and $dn$ edges in $E_1$.
\item 
For each edge $\set{v_i,v_j}\in E$, we do the following:
\begin{itemize}
\item
Create $d-4$ new nodes 
$u_{i,j,1}^1$,$u_{i,j,2}^1$,$\dots$,$u_{i,j,\left\lceil\frac{d-4}{2}\right\rceil}^1$ 
and 
\\
$u_{i,j,1}^2$,$u_{i,j,2}^2$,$\dots$,$u_{i,j,\left\lfloor\frac{d-4}{2}\right\rfloor}^2$ 
in $V'$.
\item
Add $\left\lceil\frac{d-2}{2}\right\rceil$ new edges 
\\
$\set{v_i^1,u_{i,j,1}^1}$,
$\set{u_{i,j,1}^1$,$u_{i,j,2}^1}$,
$\dots$,
$\set{u_{i,j,\left\lceil\frac{d-4}{2}\right\rceil-1}^1,u_{i,j,\left\lceil\frac{d-4}{2}\right\rceil}^1}$, 
$\set{u_{i,j,\left\lceil\frac{d-4}{2}\right\rceil}^1,v_j^1}$
and 
$\left\lfloor\frac{d-2}{2}\right\rfloor$ new edges 
\\
$\set{v_i^2,u_{i,j,1}^2}$,
$\set{u_{i,j,1}^2$,$u_{i,j,2}^2}$,
$\dots$, 
$\set{u_{i,j,\left\lfloor\frac{d-4}{2}\right\rfloor-1}^2,u_{i,j,\left\lfloor\frac{d-4}{2}\right\rfloor}^2}$, 
$\set{u_{i,j,\left\lfloor\frac{d-4}{2}\right\rfloor}^2,v_j^2}$
in $E_1$.
Note that these edges create a $d$-cycle involving the two edges 
$\set{v_i^1,v_i^2}$
and 
$\set{v_j^1,v_j^2}$; we refer to this cycle as an ``edge cycle''.
\end{itemize}
These steps create a total of $(d-4)m$ additional nodes in $V'$
and 
$(d-2)m$ additional edges in $E_1$.
\item
Let 
$E_2=E_1\setminus \set{\set{v_i^1,v_i^2} \,|\, 1\leq i\leq n }$.
\end{itemize}
Thus, 
$|V'|=dn+(d-4)m$, 
$|E_1|=dn+(d-2)m$
and 
$|E_2|=(d-1)n+(d-2)m$.
To verify that the reduction is possible for any $d$ in the range of values as claimed in the theorem, 
note that 
\[
d \leq o(|V'|)
\,\equiv\,
{d}/{|V'|} \leq o(1)
\,\Leftarrow\,
n^{-1} \leq o(1)
\]
and the last inequality is trivially true.
By \textbf{($\star_{\text{\mnc}}$)} and \textbf{($\star\star_{\text{\mnc}}$)},
the proof is complete once we prove the following lemma.

\begin{lemma}\label{ll11}
$\optmnc(G)=\mathsf{OPT}_{\mbox{\badp}_{\mathfrak{C}^2_d}}(G_1,G_2)$.
\end{lemma}

\begin{proof}
Let $E_d=E_1\setminus E_2$.
Let $f$ be the total number of cycles of at most $d$ edges in $G_1$; thus 
\[
\mathfrak{C}^2_d(G_1)=
|V'|-|E_1|+f = 
-2m+f
\]
Note that any cycle of at most $d$ edges containing an edge from $E_d$ \emph{must} be either a node cycle
or an edge cycle since a cycle 
containing an edge from $E_d$ that is \emph{neither} a node cycle
\emph{nor} an edge cycle has a number of edges that is at least 
$
2+2\times\left\lfloor\frac{d-2}{2}\right\rfloor
+\left\lceil\frac{d-2}{2}\right\rceil
=
d + \left\lfloor\frac{d-2}{2}\right\rfloor
>d$ since $d\geq 4$.
Since removing all the edges in $E_d$ removes \emph{every} node and \emph{every} edge cycle, 
\[
\mathfrak{C}^2_d(G_2)=
|V'|-|E_2|+ (f-n-m) = 
\big( |V'|-|E_1|+f \big) - m
=
\mathfrak{C}^2_d(G_1) - m
\]
Given an optimal solution 
$\vmnc\subset V$ of \mnc\ on $G$ 
of $\optmnc(G)$ nodes,  
consider the graph $G_3=(V',E_3)$ 
where 
$E_3=E_1\setminus E_d'$ and 
$E_d'=\set{ \set{v_i^1,v_i^2} \,|\, v_i\in \vmnc}\subseteq E_d$.
Since every edge of $G$ is incident on one or more nodes in $\vmnc$, 
\emph{every} edge cycle and \emph{exactly} $|E_d'|=\optmnc(G)$ node cycles of $G_1$ are removed in 
$G_3$, and thus 
\begin{multline*}
\mathfrak{C}^2_d(G_3)=
|V'|-(|E_1|-\optmnc(G) )+ (f-\optmnc(G)-m) 
\\
= 
\mathfrak{C}^2_d(G_1) - m
=
\mathfrak{C}^2_d(G_2)
\end{multline*}
This shows that 
$\mathsf{OPT}_{\mbox{\badp}_{\mathfrak{C}^2_d}}(G_1,G_2) \leq \optmnc(G)$.
Conversely, 
consider an optimal solution 
$E_d'\subseteq E_d$ of 
\badp$_{\mathfrak{C}^2_d}$ for $G_1$ and $G_2$, and 
let $G_3=(V',E_3)$ 
where 
$E_3=E_1\setminus E_d'$. 
Note that \emph{exactly} $|E_d'|=\mathsf{OPT}_{\mbox{\badp}_{\mathfrak{C}^2_d}}(G_1,G_2)$
node cycles of $G_1$ are removed in $G_3$. 
Let $m'$ be the number of edge cycles 
of $G_1$ removed in $G_3$. Then, 
\[
\mathfrak{C}^2_d(G_3)=
|V'|-(|E_1|-|E_d'| )+ (f-|E_d'|-m') = 
\mathfrak{C}^2_d(G_1) - m'
\]
and consequently $m'$ must be equal to $m$ to satisfy the constraint 
$\mathfrak{C}^2_d(G_3)=\mathfrak{C}^2_d(G_1) - m$, which implies that $G_3$ contains 
\emph{no} edge cycles. This implies that, for every edge cycle involving the two edges 
$\set{v_i^1,v_i^2}$ and $\set{v_j^1,v_j^2}$ in $G_1$, at least 
one of these two edges must be in $E_d''$, which in turn implies that 
the set of nodes 
$V''= \set{v_i \,|\, \set{v_i^1,v_i^2} \in E_d' }$ in $G$ contains 
at least one of the nodes $v_i$ or $v_j$ for every edge $\set{v_i,v_j}\in E$. 
Thus, $V''$ is a valid solution of \mnc\ on $G$ and 
$\optmnc(G) \leq |V''|=|E_d''|=\mathsf{OPT}_{\mbox{\badp}_{\mathfrak{C}^2_d}}(G_1,G_2)$.
\end{proof}

\bigskip
\noindent
\textbf{Proof of (\emph{b}3)}

\medskip
We describe the proof for $d=4$ only; the proof for $d>4$ is very similar.
Suppose, for the sake of contradiction, that one \emph{can} in fact compute 
$\mathsf{OPT}_{\mbox{\badp}_{\mathfrak{C}^2_4}}(G_1,G_2)$
in $O^\ast\big(2^{o(n)}\big)$ time where each of $G_1$ and $G_2$ has $n$ nodes.
We start with an instance $\Phi$ of \tSAT having $n$ variables and $m$ clauses.
The ``sparsification lemma'' in~\cite{IPZ01} proves the following result: 
\begin{quote}
\em
for every constant $\eps>0$, there is a constant $c>0$ such that 
there exists a $O\big(2^{\,\eps n}\big)$-time algorithm that produces 
from $\Phi$ 
a set of $t$ instances $\Phi_1,\dots,\Phi_t$ of \tSAT on these $n$ variables 
with the following properties:
\begin{itemize}
\item
$t\leq 2^{\,\eps n}$,
\item
each $\Phi_j$ is an instance of \tSAT with $n_j\leq n$ variables and $m_j\leq cn$ clauses, and 
\item
$\Phi$ is satisfiable if and only if at least one of 
$\Phi_1,\dots,\Phi_t$ is satisfiable.
\end{itemize}
\end{quote}

For each such above-produced \tSAT instance $\Phi_j$, 
we now use the reduction mentioned in 
($\star\!\star\!\star_{\text{\mnc}}$)
to produce an instance 
$G_j=(V_j,E_j)$ of \mnc\ of $|V_j|=3n_j+2m_j\leq (3+2c)\,n$ nodes and 
$|E_j|=n_j+m_j\leq (1+c)\,n$ edges
such that $\Phi_j$ is satisfiable if and only if 
$\optmnc(G_j)=n_j+2m_j$.
Now, using the reduction as described in the proof of 
parts \textbf{(\emph{b}1)} and \textbf{(\emph{b}2)} of this theorem and Lemma~\ref{ll11} thereof,
we obtain an instance 
$G_{1,j}=(V_j',E_{1,j})$ and $G_{2,j}=(V_j',E_{2,j})$
of \badp$_{\mathfrak{C}^2_3}$
such that $|V_j'|=4|V_j|<(12+8c)n$.
By assumption, we can compute 
$\mathsf{OPT}_{\mbox{\badp}_{\mathfrak{C}^2_3}}(G_{1,j},G_{2,j})$
in $O^\ast\big(2^{o(n)}\big)$,
and consequently $\optmnc(G_j)$
in $O^\ast\big(2^{o(n)}\big)$ time, 
which in turn leads us to decide in 
$O^\ast\big(2^{o(n)}\big)$ time 
if $\Phi_j$ is satisfiable for every $j$.
Since 
$t\leq 2^{\,\eps n}$
for every constant $\eps>0$,
this provides a 
$O^\ast\big(2^{o(n)}\big)$-time 
algorithm for \tSAT, contradicting \ETH.

\bigskip
\noindent
\textbf{Proof of (\emph{b}4)}

\medskip
The proof is very similar to that in \textbf{(b3)} except that now 
we start with the following lower bound result on parameterized complexity (\EG, see~\cite[Theorem 14.21]{CygFKLMPPS2015}): 
\begin{quote}
{\em 
assuming \ETH to be true, if $\optmnc(G)\leq k$
then there is \emph{no} $O^\ast\big(n^{o(k)}\big)$-time algorithm for exactly computing $\optmnc(G)$.
}
\end{quote}


\subsection{Gromov-hyperbolic curvature: computational hardness of \badp$_{\curvgrom}$} 
\label{sec-badp-all}

\begin{theorem}\label{thm-grom-hard}
It is $\NP$-hard to design a polynomial-time 
algorithm to approximate 
\badp$_{\curvgrom}(G_1,G_2)$ 
within a factor of $c\,{n}$ for some constant $c>0$, 
where $n$ is the number of nodes in $G_1$ or $G_2$.
\end{theorem}

\subsubsection{Proof techniques and relevant comments regarding Theorem~\ref{thm-grom-hard}}
\label{sec-informal-badp-gromov}

The reduction is from the \emph{Hamiltonian path problem for cubic graphs} (\ham), and shown schematically in 
\FI{fig2}. 
Conceptually, the idea is to amplify the difference between Hamiltonian and non-Hamiltonian paths 
to a large size difference of ``geodesic'' triangles (\emph{cf}.\ Definition~\ref{def-hyperbolic-1})
such that application of results such as~\cite[Lemma~2.1]{RT04}
can lead to a large difference of the corresponding Gromov-hyperbolicity values.
To get the maximum possible amplification (maximum gap in lower bound) 
we need to make very careful and precise arguments regarding the Gromov-hyperbolicities of classes of graphs.
Readers should note that Gromov-hyperbolicity value is not necessarily related to the circumference of a graph,
and thus the reduction cannot rely simply on presence or absence of long paths or long cycles in the constructed graph.

The inapproximability reduction necessarily requires some nodes with large (close to linear) degrees 
even though with start with \ham in which every node has degree exactly $3$.
We conjecture that our large inapproximability bounds do \emph{not} hold when the given graphs have nodes of 
bounded degree, but have been unable to prove it so far.

\begin{figure}[htbp]
\centerline{\includegraphics[scale=0.66]{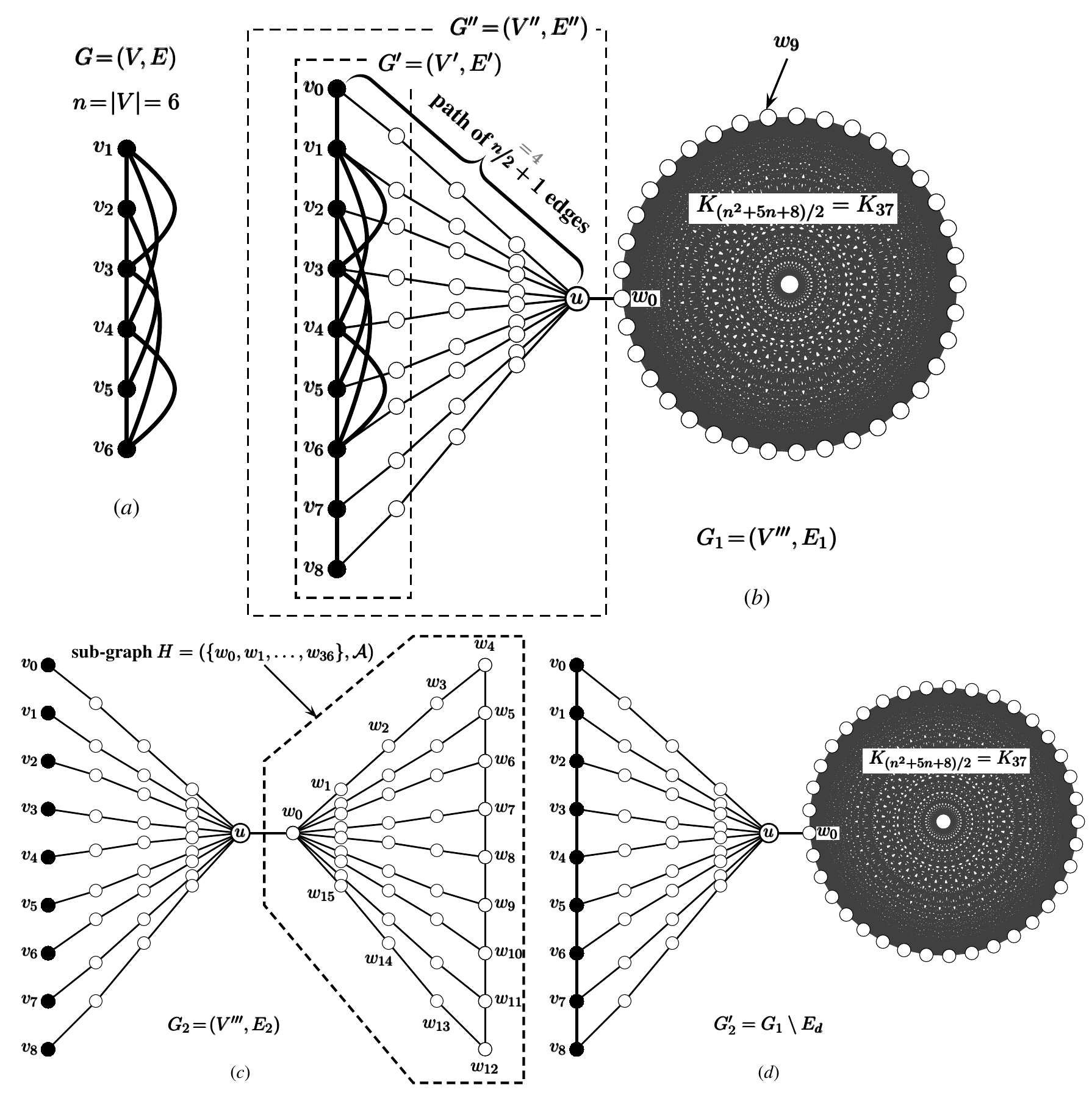}}
\caption{\label{fig2}Illustration of the reduction in Theorem~\ref{thm-grom-hard}.
(\emph{a}) 
The input graph $G=(V,E)$ for the 
Hamiltonian path problem for cubic graphs (\ham).
(\emph{b}) and (\emph{c}) The graphs $G_1=(V'',E_1)$ and $G_2=(V'',E_2)$ for the generated instance of \badp$_{\curvgrom}(G_1,G_2)$.
The graph $G'=(V',E')$ obtained from the given graph $G'$ by adding three extra nodes and three extra edges.
(\emph{d}) An optimal solution $G_2'$ for \badp$_{\curvgrom}(G_1,G_2)$
if $G$ contains a Hamiltonian path between $v_1$ and $v_n$.
}
\end{figure}


\subsubsection{Proof of Theorem~\ref{thm-grom-hard}}
\label{sec-app-proof-thm-grom-hard}

We will prove our inapproximability result via a reduction from the 
\emph{Hamiltonian path} problem for cubic graphs (\ham)
which is defined as follows.

\begin{definition}[Hamiltonian path problem for cubic graphs (\ham)]
Given a cubic (\IE, a $3$-regular) graph $G=(V,E)$ and two specified nodes $u,v\in V$, does $G$ 
contain a Hamiltonian path between $u$ and $v$, \IE, a path between $u$ and $v$ that visits every node
of $G$ exactly once?
\end{definition}

\ham is known to be $\NP$-complete~\cite{GJT76}. 
Consider an instance $G=(V,E)$ and $v_1,v_n\in V$ 
of \ham of $n$ nodes and $m=3n/2$ edges 
where $V=\set{v_1,v_2,\dots, v_n}$, $E=\set{e_1,e_2,\dots,e_m}$ and the goal is to determine if there
is a Hamiltonian path between $v_1$ and $v_n$
(see \FI{fig2}(\emph{a})).
We first introduce three new nodes $v_0$, $v_{n+1}$ and $v_{n+2}$, and connect them to the nodes in $G$
by adding three new edges 
$\set{v_0,v_1}$, 
$\set{v_n,v_{n+1}}$ and 
$\set{v_{n+1},v_{n+2}}$,
resulting in the graph 
$G'=(V',E')$
(see \FI{fig2}(\emph{b})).
It is then trivial to observe the following:
\begin{itemize}
\item
$G$ has a Hamiltonian path between $v_1$ and $v_n$ 
if and only if 
$G'$ has a Hamiltonian path between $v_0$ and $v_{n+2}$.  
\item
If $G'$ does have a Hamiltonian path then such a path must be between the two nodes 
$v_0$ and $v_{n+2}$.  
\end{itemize}
Note that 
$|V'|=n+3$ and $|E'|=(3n/2)+3$.
We next create the graph $G''=(V'',E'')$ from $G'$ in the following manner
(see \FI{fig2}(\emph{b})):
\begin{itemize}
\item
We add a set of $1+{(n^2+3n)}/{2}$ new nodes 
$u,v_{0,1}$,$\dots$,$v_{0,{n}/{2}}$,
$v_{1,1}$,$\dots$,$v_{1,{n}/{2}}$,
$\dots$,
$v_{n+2,1}$,$\dots$,$v_{n+2,{n}/{2}}
$.
For notational convenience, we set 
$u\eqdef v_{i,0}$ for all $i\in\{0,1,\dots,n+2\}$
and 
$v_j\eqdef v_{j,({n}/{2})+1}$ for all $j\in\{0,1,\dots,n+2\}$.
\item
We add a set of $n+3$ disjoint paths (each of length $\frac{n}{2}+1$) 
$\cP_0,\cP_1,\dots,\cP_{n+2}$ where 
$
\cP_j \eqdef 
v_{j,0} \leftrightarrow v_{j,1} \leftrightarrow v_{j,2} \leftrightarrow \dots \leftrightarrow v_{j,\frac{n}{2}+1}$.
\end{itemize}
Note that 
$|V''|=n+4+\frac{n^2+3n}{2}=\frac{n^2+5n}{2}+4$ and $|E''|=\frac{3n}{2}+3+(n+3)\left(\frac{n}{2}+1\right)=\frac{n^2}{2}+4n+6$.
We now create an instance
$G_1=(V,E_1)$ and $G_2=(V,E_2)$ (with $\emptyset\subset E_2\subset E_1$)
of \badp$_{\curvgrom}(G_1,G_2)$ from $G''$ in the following manner
(see \FI{fig2}(\emph{b})--(\emph{c})): 
\begin{itemize}
\item 
The graph $G_1=(V''',E_1)$ is obtained by modifying $G''$ as follows: 
\begin{itemize}
\item 
Add a complete graph 
$K_{|V''|}$ on $|V''|=\frac{n^2+5n}{2}+4$ nodes $w_0$,$w_1$,$\dots$,$w_{|V''|-1}$
and the edge 
$\set{u,w_0}$.
This step adds 
$|V''|$ new nodes 
and 
$\binom{|V''|}{2}+1$ new edges. 
\end{itemize}
Thus, we have 
$|V'''|=2\,|V''|=n^2+5n+8$,
and 
\[
|E_1|=|E''|+\binom{|V''|}{2}+1
=\frac
{ n^4 + 10n^3 + 43n^2+102n+104 }
{8}
\]
\item 
The graph $G_2=(V''',E_2)$ is obtained from $G_1$ as follows. Let 
$\cA$ be the set of edges of a sub-graph of the graph
$K_{|V''|}$ (added in the previous step)
that is isomorphic to the graph 
$(V'',\widehat{E})$ where 
\[
\widehat{E}= \big( E'' \setminus E \big) \bigcup \big\{ \set{v_j,v_{j+1}} \,|\,j\in \{0,1,\dots,n+1\} \big\} 
\]
and the node $w_0$ is mapped to the node $u$ in the isomorphism.
Such a sub-graph can be trivially found in polynomial time. 
For notational convenience we number the nodes in this sub-graph such that the order of 
the nodes in the largest cycle (having $2n+4$ edges) of this sub-graph 
is $w_0,w_1,\dots,w_{2n+3}$
(see \FI{fig2}(\emph{c})).
We then set $E_2=E''\cup A \cup \set{u,w_0}$.
Thus,
\begin{multline*}
|E_2|=|E''|+|\cA|+1
=
|E''|+|\widehat{E}|+1
\\
=
|E''|+ \big( |E''| - |E| + n + 2 \big) +1
=
n^2+\frac{15n}{2}+12
\end{multline*}
%
\end{itemize}
%
We first need to prove some bounds on the hyperbolicities of various
graphs and sub-graphs that appear in our reduction.
It is trivial to see that 
$\curvgrom(K_{|V''|})=0$.
Define 
$\widetilde{\Delta_{u,v,w}}(G)$
be a geodesic triangle which contributes to the \emph{minimality} of the value of 
$\curvgrom(G)$, \IE, 
one of the shortest paths, say $\overline{u,v}$, lies in a $\curvgrom(\widetilde{\Delta_{u,v,w}}(G))$-neighborhood
of the union 
$\overline{u,w}\,\cup\,\overline{v,w}$
of the other two shortest paths,  
but 
$\overline{u,v}$ does \emph{not} lie in a $\delta$-neighborhood
of $\overline{u,w}\,\cup\,\overline{v,w}$ for any $\delta<\curvgrom(\widetilde{\Delta_{u,v,w}}(G))$.
The following two facts are well-known.
%

\begin{fact}\label{fact1}
For any geodesic triangle $\Delta_{u,v,w}$, 
from the definition of $\curvgrom(\Delta_{u,v,w})$ \emph{(\emph{cf}.\ Definition~\ref{def-hyperbolic-1})} it follows that
\[
\curvgrom(\Delta_{u,v,w})\leq \max\left\{
\left  \lfloor \nicefrac{\dist_G(u,v)}{2} \right \rfloor,
\left  \lfloor \nicefrac{\dist_G(v,w)}{2} \right \rfloor,
\left  \lfloor \nicefrac{\dist_G(u,w)}{2} \right \rfloor
\right\}
\]
\end{fact}

\begin{fact}[{\cite[Lemma~2.1]{RT04}}]\label{fact2}
We may assume that 
$\widetilde{\Delta_{u,v,w}}(G)$
is a \emph{simple} geodesic triangle, \IE, 
the three shortest paths 
$\overline{u,v}$, $\overline{u,w}$ and $\overline{v,w}$
do not share any nodes other than $u$, $v$ or $w$.
\end{fact}
%

Let $H$ denote the (node-induced) sub-graph $(\{w_0,w_1,\dots,w_{|V''|-1}\},\cA)$ of $G_2$.

\begin{lemma}\label{cycle-hyp}
$\curvgrom(G_2)=\curvgrom(H)=\frac{n}{2}+1$.
\end{lemma}

\begin{proof}
By Fact~\ref{fact2}
$\widetilde{\Delta_{p,q,r}}(G_2)$
must be a simple geodesic triangle and therefore 
can only include edges in $\cA$. 
Since the diameter of the sub-graph 
$H$
is $n+2$, 
for any geodesic triangle 
$\Delta_{p,q,r}$ of 
$H$
we have 
\[
\max\left\{
\left  \lfloor \nicefrac{\dist_G(p,q)}{2} \right \rfloor,
\left  \lfloor \nicefrac{\dist_G(q,r)}{2} \right \rfloor,
\left  \lfloor \nicefrac{\dist_G(p,r)}{2} \right \rfloor
\right\}
\leq n+2 
\]
and thus by Fact~\ref{fact1} we have 
$
\curvgrom(G_2)=
\curvgrom(H)=
\curvgrom(\widetilde{\Delta_{p,q,r}}(G_2))\leq 
\frac{n}{2}+1
$.
Thus, it suffices we provide a simple geodesic triangle 
$\Delta_{p,q,r}$ of $H$ for some three nodes $p,q,r$ of $H$ such that 
$\curvgrom({\Delta_{p,q,r}}(H))=\frac{n}{2}+1$. 
Consider the simple geodesic triangle 
$\Delta_{w_0,w_{\frac{n}{2}+1},w_{\frac{3n}{2}+3}}$ of $H$ consisting of the three shortest paths 
$\cQ_1\eqdef w_0\leftrightarrow w_1\leftrightarrow w_2\leftrightarrow\dots\leftrightarrow 
w_{\frac{n}{2}}\leftrightarrow w_{\frac{n}{2}+1}$, 
$\cQ_2\eqdef w_{\frac{n}{2}+1}\leftrightarrow w_{\frac{n}{2}+2}\leftrightarrow w_{\frac{n}{2}+3}
\leftrightarrow\dots\leftrightarrow w_{\frac{3n}{2}+2}\leftrightarrow w_{\frac{3n}{2}+3}$
and 
$\cQ_3\eqdef w_{\frac{3n}{2}+3}\leftrightarrow w_{\frac{3n}{2}+4}\leftrightarrow w_{\frac{3n}{2}+5}\leftrightarrow 
\dots\leftrightarrow w_{2n+3},w_0$,
and 
consider the node $w_{n+2}$ that is the mid-point of the shortest path 
$\cQ_2$
(see \FI{fig2}(\emph{c})).
It is easy to verify that the distance of the node 
$w_{n+2}$
from the union of the two shortest paths 
$\cQ_1$
and 
$\cQ_3$
is
$\frac{n}{2}+1$.
\end{proof}


Now, suppose that we can prove the following two claims:
\[
\begin{array}{r l}
\text{\bf (completeness)} &
\text{if $G$ has a Hamiltonian path between $v_1$ and $v_n$ then}
\\
\\
[-10pt]
& \hspace*{1in}
     \mathsf{OPT}_{\mbox{\badp}_{\curvgrom}}(G_1,G_2)\leq \frac{n}{2}+1
\\
\\
[-10pt]
\text{\bf (soundness)} &
\text{if $G$ has no Hamiltonian paths between $v_1$ and $v_n$ then}
\\
\\
[-10pt]
& \hspace*{1in}
     \mathsf{OPT}_{\mbox{\badp}_{\curvgrom}}(G_1,G_2)\geq \frac{n^3+3n^2+2n}{2}
\end{array}
\]
Note that this proves the theorem since 
$
\frac { \frac{n^3+3n^2+2n}{2} } { \frac{n}{2}+1 }
> \frac{n^2}{5} =\Omega \left( {|V'''|} \right)
$.

\bigskip
\noindent
\textbf{Proof of completeness}

\medskip
\noindent
Suppose that $G$ has a Hamiltonian path between $v_1$ and $v_n$, say 
$v_1\leftrightarrow v_2\leftrightarrow v_3\leftrightarrow\dots\leftrightarrow v_{n-1}\leftrightarrow v_n$.
Thus, $G''$ has a Hamiltonian path 
$v_0\leftrightarrow v_1\leftrightarrow v_2\leftrightarrow v_3\leftrightarrow\dots\leftrightarrow v_{n-1}\leftrightarrow 
v_n\leftrightarrow v_{n+1}\leftrightarrow v_{n+2}$
between $v_0$ and $v_{n+2}$.
We remove the $\frac{n}{2}+1$ edges in $E_d=E''\setminus \big\{ \,\set{v_j,v_{j+1}} \,|\, j=0,1,\dots,n+1 \big\}$
that are not in this Hamiltonian path resulting in the graph 
$G_2'=G_1\setminus E_d$
(see \FI{fig2}(\emph{d})).
To show that 
$\curvgrom(G_2')=\curvgrom(G_2)$, 
note that 
by Fact~\ref{fact2}
$\widetilde{\Delta_{p,q,r}}(G_2')$
must be a simple geodesic triangle and therefore 
\begin{multline*}
\curvgrom(G_2')= \max \big\{ \, \curvgrom(G''\setminus E_d),\,  \curvgrom(K_{|V''|})\, \big\}
\\
=
\max \big\{ \curvgrom(G''\setminus E_d),\, 0 \big\}
=
\curvgrom(G''\setminus E_d)
\end{multline*}
Since $G''\setminus E_d$ is isomorphic to $H$,  
by Lemma~\ref{cycle-hyp}
$
\curvgrom(G''\setminus E_d)
=
\curvgrom(H)
=
\curvgrom(G_2)
$.

\bigskip
\noindent
\textbf{Proof of soundness}

\medskip
\noindent
Assume that $G$ has no Hamiltonian paths between $v_1$ and $v_n$, and 
let $E_d\subseteq E_1\setminus E_2$ be the optimal set of edges that need to be deleted
to obtain the graph $G_2'=(V''',E_1\setminus E_d)$ such that 
$\curvgrom(G_2')=\curvgrom(G)$. 
By Fact~\ref{fact2},
$\widetilde{\Delta_{p,q,r}}(G_2')$
must be a simple geodesic triangle and therefore 
\begin{multline}
\curvgrom(G_2')= 
\max \big\{ \, \curvgrom(G''\setminus E_d),\,  \curvgrom(K_{|V''|}\setminus E_d)\, \big\}
\\
=
\curvgrom(G_2)=\frac{n}{2}+1 
\label{eq1}
\end{multline}

\begin{lemma}\label{lem-simpler}
$\curvgrom(G''\setminus E_d)\leq \frac{n}{2}$.
\end{lemma}

\begin{proof}
Since $G$ has no Hamiltonian paths between $v_1$ and $v_n$,
$\diam(G'\setminus E_d)\leq n+1$. 
Assume, for the sake of contradiction, that 
$\curvgrom(G''\setminus E_d)\geq \frac{n}{2}+1$.
By Fact~\ref{fact1}, we have 
\begin{multline*}
\curvgrom(G''\setminus E_d)
=
\curvgrom( \widetilde{\Delta_{p,q,r}}(G''\setminus E_d) )
\\
\leq
\max\left\{
\left  \lfloor \nicefrac{\dist_{G''\setminus E_d}(p,q)}{2} \right \rfloor,
\left  \lfloor \nicefrac{\dist_{G''\setminus E_d}(q,r)}{2} \right \rfloor,
\left  \lfloor \nicefrac{\dist_{G''\setminus E_d}(p,r)}{2} \right \rfloor
\right\}
\end{multline*}
and thus at least one of the three distances in the left-hand-side of the above inequality, 
say $\dist_{G''\setminus E_d}(p,q)$,
must be at least $n+2$.
Let $\cL(\cC(H))$ and $\cL(H)$ denote the length (number of edges) of a (simple) cycle $\cC$ and the length 
of the \emph{longest (simple) cycle} of a graph $H$. 
Since 
$\curvgrom(\widetilde{\Delta_{p,q,r}}(G''\setminus E_d))>0$
and 
$\widetilde{\Delta_{p,q,r}}(G''\setminus E_d)$
must be a simple geodesic triangle, 
there must be at least one cycle, say $\cC$, in $G''\setminus E_d$ containing $p$, $q$ and $r$.
Now, note that 
\[
\cL(\cC(G''\setminus E_d))
\leq
\cL(G''\setminus E_d)
\leq
2\,\left(\frac{n}{2}+1\right) + \diam(G'\setminus E_d)
\leq 
2n+3
\]
and therefore 
$\dist_{G''\setminus E_d}(p,q)\leq \left\lfloor \frac{2n+3}{2}\right\rfloor=n+1$, which provides 
the desired contradiction.
\end{proof}

By Lemma~\ref{lem-simpler} and Equation~\eqref{eq1} it follows that 
$\curvgrom(K_{|V''|}\setminus E_d)=\frac{n}{2}+1$.

\begin{lemma}
If 
$\curvgrom(K_{|V''|}\setminus E_d)\geq \frac{n}{2}+1$ then 
$|E_d|\geq \frac{n^3+3n^2+2n}{2}$.
\end{lemma}

\begin{proof}
Since 
$
\curvgrom(K_{|V''|}\setminus E_d)
=
\curvgrom( \widetilde{\Delta_{p,q,r}}(K_{|V''|}\setminus E_d) )
=
\frac{n}{2}+1
$, 
by Fact~\ref{fact1}
at least one of the three distances 
$\dist_{K_{|V''|}\setminus E_d}(p,q)$, 
$\dist_{K_{|V''|}\setminus E_d}(q,r)$ 
or 
$\dist_{K_{|V''|}\setminus E_d}(p,r)$,
say $\dist_{K_{|V''|}\setminus E_d}(p,q)$,
must be at least $n+2$.
This implies that 
$K_{|V''|}\setminus E_d$ 
must contain a shortest path of length $n+2$, say 
$\cQ\eqdef w_0\leftrightarrow w_1\leftrightarrow w_2\leftrightarrow\dots\leftrightarrow w_{n+1}\leftrightarrow w_{n+2}$.
We now claim that \emph{no} node from the set 
$W_1=\set{w_{n+3},w_{n+4},\dots,w_{|V''|-1}}$
is connected to \emph{more than} $3$ nodes from the set 
$W_2=\set{w_0,w_1,\dots,w_{n+2}}$
in $K_{|V''|}\setminus E_d$.
To show this by contradiction, suppose that some node 
$w_i\in W_1$ is connected to four nodes $w_j,w_k,w_\ell,w_r\in W_2$ with $j<k<\ell<r$.
Then $r\geq j+3$ which implies 
$\dist_{K_{|V''|}\setminus E_d}(w_j,w_r)\leq 2$, contradicting the fact that $\cQ$ is a shortest path. 
It thus follows that 
\begin{multline*}
|E_d|\geq 
\big( (n+3)-3\big)|W_1| 
=
n (|V''|-(n+3))
\\
=
n \left(\frac{n^2+3n}{2}+1\right)
=
\frac{n^3+3n^2+2n}{2}
\end{multline*}
\end{proof}

The above lemma completes the proof of soundness of our reduction.



\section{Conclusion and future research}
\label{sec-conclusion}

Notions of curvatures of higher-dimensional geometric shapes and topological spaces 
play a \emph{fundamental} role in physics and mathematics 
in characterizing anomalous behaviours of these higher dimensional entities. 
However, using curvature measures to detect anomalies in networks 
is \emph{not} yet very common due to several reasons such as lack of preferred geometric interpretation of networks and lack of 
experimental evidences that may lead to specific desired curvature properties.
In this paper we have attempted to formulate and analyze curvature analysis methods to provide the foundations of 
systematic approaches to find critical components and anomaly detection in networks by using 
two measures of network curvatures, namely the Gromov-hyperbolic curvature
and the geometric curvature measure.
This paper must \emph{not} be viewed as uttering the final word 
on appropriateness and suitability of specific curvature measures, but rather
should be viewed as a stimulator and motivator of further theoretical or empirical research
on the exciting interplay between notions of curvatures from network and non-network domains.

There is a \emph{plethora} of interesting future research questions and directions raised by the topical discussions 
and results in this paper. Some of these are stated below.
\begin{enumerate}[label=$\triangleright$,leftmargin=0.7cm]
\item
For geometric curvatures, we considered the first-order non-trivial measure 
$\mathfrak{C}^2_d$. 
It would be of interest to investigate computational complexity issues of anomaly detection 
problems using $\mathfrak{C}^p_d$ for $p>2$. We conjecture that our algorithmic results 
for extremal anomaly detection using 
$\mathfrak{C}^2_d$ (Theorem~\ref{ext-thm}(\emph{a}2-2)\&(\emph{b}2))
can be extended to 
$\mathfrak{C}^3_d$.
\item
There are at least two more aspects of geometric curvatures that need further careful investigation.
Firstly, the topological association of elementary components to higher-dimensional objects as
described in this paper is by \emph{no} means the only reasonable topological association possible.
But, more importantly, 
other suitable notions of geometric curvatures are quite possible. 
As a very simple illustration, 
assuming that 
smaller dimensional simplexes 
edges 
in the discrete network setting correspond to vectors or directions in the smooth context, 
an analogue of the Bochner-Weitzenb\"{o}ck formula developed by 
Forman for the curvature for 
a simplex $s$ can be 
given by the formula~\cite{Fo03,SSG18}:
\[
\mathfrak{F}(s)
=
w_s 
\left(
          \left(
               \sum_{s \prec s'} \frac{ w_s } { w_{s'} }
             + \sum_{s' \prec s} \frac{ w_{s'} } { w_s }
			   \right)	
 \!\!
 \text{--}
 \!\!
   \sum_{ s' \parallel s  }
             \Big| 
 \!\!
							 \sum_{ s,s' \prec g  }
 \!\!
						     \frac  { \sqrt{ w_s w_{s'} } } {  w_g }
							+
 \!\!
							  \sum_{ g \prec s,s' }
						   \frac 
							   { w_g } { \sqrt { w_s w_{s'} } }  
						 \Big|
\right)
\]
where
$a\prec b$ means $a$ is a face of $b$, 
$a\parallel b$ means 
$a$ and $b$ have either a common higher-dimensional face or a common lower-dimensional face 
but \emph{not} both,
and $w$ is a function that assigns weights to simplexes.
One can then either modify the Euler characteristics as 
$\sum_{k=0}^p (-1)^k\, \mathfrak{F}(f_d^k)$ 
or 
by combining the individual $\mathfrak{F}(f_d^k)$ values 
using curvature functions defined by Bloch~\cite{Bl14}. 
\item
Our inapproximability results for the Gromov-hyperbolic curvature require a 
high average node degree. 
We hypothesize that the anomaly detection 
problems using Gromov-hyperbolic curvatures is much more computationally tractable than 
what our results depict for networks with bounded average degree.
\end{enumerate}
Finally, in contrast to the combinatorial/geometric graph-property based approach investigated in this paper,
a viable alternate approach for anomaly detection is the \emph{algebraic tensor-decomposition based approach} studied 
in the contexts of dynamic social networks~\cite{STF06} and 
pathway reconstructions in cellular systems and microarray data integration from several 
sources~\cite{AG05,OGA07}.
This approach is quite different from the ones studied in this paper with its own pros and cons.
For computational biology researchers, an useful survey of tensor-based approaches for various kinds 
of biological networks and systems can be found in reference~\cite{YDA20}.

\section*{Acknowledgements}

We thank Anastasios Sidiropoulos and Nasim Mobasheri for very useful discussions. This research work was partially supported by 
NSF grants IIS-1160995 and IIS-1814931.



\end{document}